\documentclass[a4paper,10pt]{article}
\usepackage{fullpage}

\usepackage[utf8]{inputenc}
\usepackage[english]{babel}

\usepackage{amsmath,amsfonts,amsthm,amssymb,amsxtra,enumerate,mathtools,mathrsfs}
\usepackage{fullpage}

\usepackage{graphics}
\usepackage{comment}
\usepackage{enumitem}
\usepackage{url}
\usepackage{faktor} 
\usepackage{yfonts} 

\usepackage{tikz}
\usepackage{tikz-cd}
\usepackage{pgfplots}
\usepackage{graphicx}
\graphicspath{ {./pic/} }

\usepackage[all]{xy}

\usepackage{palatino}
\usepackage{mathpazo}
\usepackage{bm}

\usepackage{array}

\usepackage{tikz}
\usetikzlibrary{positioning,arrows.meta}

\usepackage{totcount}

\usepackage{tabularx}
\usepackage{booktabs}

\usepackage{todonotes}
\usepackage{etoolbox}
\usepackage{hyperref}
\usepackage{cleveref}

\makeatletter
\DeclareRobustCommand{\cev}[1]{%
  {\mathpalette\do@cev{#1}}%
}
\newcommand{\do@cev}[2]{%
  \vbox{\offinterlineskip
    \sbox\z@{$\m@th#1 x$}%
    \ialign{##\cr
      \hidewidth\reflectbox{$\m@th#1\vec{}\mkern4mu$}\hidewidth\cr
      \noalign{\kern-\ht\z@}
      $\m@th#1#2$\cr
    }%
  }%
}

\newcommand{\Cb}{\mathbb{C}}

\newcommand{\Nb}{\mathbb{N}}
\newcommand{\Rb}{\mathbb{R}}
\newcommand{\Zb}{\mathbb{Z}}
\newcommand{\Sb}{\mathbb{S}}
\newcommand{\M}{\mathscr{M}}

\newcommand{\ii}{\mathrm{i}}
\newcommand{\ee}{\mathrm{e}}
\newcommand{\dd}{\,\mathrm{d}}
\renewcommand{\phi}{\varphi}
\renewcommand{\theta}{\vartheta}
\newcommand{\p}{\partial}

\renewcommand{\Re}{\,\mathrm{Re} \,}
\newcommand{\supp}{\mathrm{supp}\,}

\newcommand{\ind}{\mathrm{ind}}

\newcommand{\dom}{\mathrm{dom}}
\newcommand{\APS}{\mathrm{APS}}

\newcommand{\ie}{\emph{i.e.}~ }
\newcommand{\eg}{\emph{e.g.}~ }
\newcommand{\cf}{\emph{cf.} }

\newcommand{\floor}[1]{\left\lfloor #1 \right\rfloor}

\newcommand{\demph}[1]{\textbf{\textup{#1}}}

\newtheorem{theorem}{Theorem}
\newtheorem{definition}[theorem]{Definition}
\newtheorem{remark}[theorem]{Remark}
\newtheorem{example}[theorem]{Example}

\newtheorem{prop}[theorem]{Proposition}
\newtheorem{lemma}[theorem]{Lemma}
\newtheorem{corollary}[theorem]{Corollary}

\newtheorem{notation}[theorem]{Notation}
\begin{document}

\clearpage

%

\begin{center}
{\LARGE\textbf{Aharonov--Casher theorems for Dirac operators on manifolds with boundary and APS boundary condition} \par} 
\medskip 
{\normalsize M.~Fialov\'{a}} 
\footnote{ISTA, Klosterneuburg, Austria. email: \emph{mfialova@ist.ac.at}, ORCID: $0000-0002-2375-7263$}
\medskip \\
{\small 15 May 2023}
\end{center}


\bigskip

\begin{abstract}
	The Aharonov--Casher theorem is a result on the number of the so-called zero modes of a system described by 
	the magnetic Pauli operator in $\Rb^2$.
	In this paper we address the same question for the Dirac operator on a flat two-dimensional manifold with boundary and 
	Atiyah--Patodi--Singer boundary condition. 
	More concretely we are interested in the plane and a disc with a finite number of circular holes  cut out.
	We consider a smooth compactly supported magnetic field on the manifold and an arbitrary magnetic field inside the holes.
\end{abstract}
\bigskip 
\textbf{key words:} Dirac operator, APS boundary condition, 
APS index theorem, Spin$^c$ spinor bundle
\bigskip

\setcounter{tocdepth}{2}
\tableofcontents
\vfill

\clearpage

\vfill

\newpage
\label{page:notation}
\section*{Notation}
\bigskip

\begin{tabular}[c]{l l}
	$\Omega_k$  	& Open ball in $\Cb$ with centre at $w_k\in \Cb$ and radius $R_k$,	\\
	$(r_k, \phi_k)$	&  Polar coordinates around the point $w_k$, 	
					 we set $\phi_k = 0$   to be the axis parallel with the \\
					 & Cartesian positive $x$-axis	\\
	$Cl(V)$ 		 & Clifford algebra on a vector space $V$ 	\\
	$C_0^{\infty}(X)$ 	& Smooth  functions with compact support in $X$	\\
	$const$ 		 & A general constant which can be of 
						different value from one  
					 	(in)equality sign to 		\\
					 & another	\\	 
	$int \, \gamma$ 	& Interior of a curve $\gamma$	\\
	$(\cdot, \cdot)_E$ 	& Inner product on fibres of a bundle $E$ 	\\
	$\Gamma(M, E)$ 	& Smooth sections of a bundle $E$ over a manifold $M$ 	\\
	$L^2(M, E) $ 
				& Square integrable sections of the bundle $E$ over a Riemannian manifold $M$ \\		
	$L^2(M,g;\Cb^2) $ 
				& $\Cb^2$-valued square integrable functions on a Riemannian manifold $M$ with metric $g$ \\	
	$M^{\circ}$		& Interior of a manifold $M$ \\		
	$TM$			& Tangent space of a manifold $M$ 	\\
	$T ^{\ast}M$	& Cotangent space of a manifold $M$	\\
	$T_p^{\ast}M$	& Fibre of the cotangent space above the point $p\in M$ \\
	$(\cdot)^T$	& vector transposition 	\\
	$\p X$ 		& Boundary of a region $X$	\\
	$\overline{X}$ 	& Closure of a subset $X \subset \Cb$	\\
	$\floor{y}$		& The biggest  integer strictly less than $y \in \Rb$	\\
	$\|\cdot\|_{S}$	& The norm on a space $S$ \\
\end{tabular}

\section{Introduction} \label{sec:Introduction}
This paper is inspired by work of Aharonov and Casher (AC), \cite{AC79}, discussing the number of zero modes
(\ie of eigenfunctions corresponding to the zero eigenvalue)
of the Dirac operator with magnetic field.
 We extend their result on $\Rb^2$
 to a plane with holes, Thm.~\ref{thm:unbdd}, 
 a disc with holes, Thm.~\ref{thm:bdd},
 and finally to a sphere with holes, Thm.~\ref{thm:sphere}.
 All our results are concerning a particular choice of the extension of the Dirac operator, given by the Atiyah--Patodi--Singer (APS) boundary condition.
 In \cite{AC79} the zero modes have a definite chirality, 
 which depends on the sign of the flux $\Phi$ of the magnetic field.
 The number of zero modes then depends on the magnitude of the flux, namely it is $\floor{\frac{|\Phi|}{2\pi}}$. 
 On the plane with holes we reproduce the same result. 
 We would like to point out, that considering a self-adjoint realisation $D$
	of the Dirac operator,
	the zero modes of $D$ coincide with the zero modes of its square
	$D^2$ which we may refer to as an operator of Pauli type.
	Such operators describe the non-relativistic limit of Dirac operators and
	due to its positivity the zero energy states 
	are also its ground states.

 The analytic index of the Dirac operator 
 is a closely related quantity as it computes 
 the difference of the number of zero modes with positive and negative   
 chirality
  \footnote{
A reader that is not familiar with the notion of positive (negative) chirality can, for purposes of the results in this paper, think of the eigenvectors of the third Pauli matrix  which is the diagonal matrix $\sigma_3 \coloneqq \mathrm{diag} (1, -1)$, corresponding to eigenvalue $+1$ ($-1$). The concept will be more generally introduced in Sec.~\ref{sec:geometry}.
}.
These are also referred to as modes with spin up or spin down.
Atiyah and Singer proved in~\cite{AS63} that the analytic index is equal to the topological index
of the underlying closed manifold.
Using the stereographic projection, the AC theorem can be reformulated as
a result on a sphere (see \eg \cite[Thm.~8.3.]{ES01}).
Since the sphere and the disc with holes are compact manifolds we 
also have the index theorem for them. It gives the formula for the difference between the number of the zero modes with positive and negative chirality. Our adaptation of the AC theorem then gives each of these numbers separately and is thus a stronger result. Although, one should of course keep in mind that the index theorem is valid for a very general setting.

A continuation of Atiyah and Singer's work resulted in the generalization 
of the index formula for manifolds with boundary by Atiyah, Patodi and Singer
in the series of papers \cite{APS1, APS2, APS3}.
The authors introduced a boundary condition, nowadays known as the APS boundary condition, which we adopted here for the definition of the domain of our Dirac operator.
It is a global boundary condition based on preservation of chirality upon reflection on the boundary. 
A manifold $X$ that has a product structure near the boundary can be extended by an infinite cylinder.
From the analytical point of view the APS boundary condition is tailored so that any zero mode of the Dirac operator on $X$ satisfying this condition can be extended to this infinite cylinder as a square integrable function plus a function constant along the infinite direction of the cylinder. For more details see \cite[Sec.~22E]{BW93}.

The literature on zero modes is vast and we will mention only a couple
of works generalizing the AC theorem.
A proof of the result on a two sphere is due to Avron and Tomaras (but was not published) 
and it can be found e.g. in \cite{CFKS87} or \cite[Appx.~A.3]{ES01}.
For generalization to measure-valued magnetic fields see~\cite{EV01}.
Singular Aharonov--Bohm type fields were considered 
by Hirokawa and Ogurisu in \cite{HO01}, by Persson in \cite{Per06} and
 by Geyler and \v{S}\v{t}ov\'{\i}\v{c}ek in \cite{GS04}.
 Rozenblum and Shirokov, \cite{RS06}, showed that for certain singular magnetic fields there could be an infinite
dimensional space of zero modes with having possibly both spin up and spin down modes.
 Results for the case of even dimensional Euclidean spaces were discussed by Persson in \cite{Per08}.
Bony, Espinoza and Raikov investigate almost periodic potentials in \cite{BER19}.
On a bounded domain with Dirichlet boundary condition the related result was studied
in \cite{Elton16}.

\bigskip 

The aim of this paper is to extend the Aharonov--Casher theorem exactly to these cases when a compact boundary is present. 
Since the APS boundary condition is a condition forged for the index theorem it seems to be a reasonable candidate to start with when  studying the AC formula (due to the relation between the AC and index formulas mentioned above) for Dirac operators on manifolds with boundary. Of course, there are many other choices of boundary conditions that would make the Dirac operator self-adjoint 
(that are not discussed here) and the validity of the AC formula is heavily dependent on the realisation we use.
For a detailed classification of 
self-adjoint realisations of Dirac type operators on manifolds with boundary we refer an interested reader to 
works \cite{BB,BB12}.

Our main motivation to study this problem is the mathematical interest to see how could the AC theorem be influenced by a presence of a boundary. We also find it a curious problem to explore the zero modes on a non-compact manifold, where the index theorem is not applicable.
The particular setting of a plane or a ball with holes is of interest also due to the Aharonov--Bohm (AB) effect. This phenomenon gives a possibility to observe magnetic field in quantum mechanics even if the field is supported in a region inaccessible to the particle. The net observable effect then depends only on the flux of the field in this region. In our setting of magnetic field supported in the holes the AC formula precisely demonstrates such properties. Let us mention, that the AB effect is often studied using the model of an infinitesimally thin and infinitely long solenoid. This corresponds to the magnetic field formally given as $\alpha \delta_0$, where $\alpha$ is the magnetic flux and $\delta_0$ the delta distribution with support at $\{0\}$.
There are many works considering the Schr\"{o}dinger-type operators with such a point interaction (also referred to as the AB field). 
The domain of a realisation is then characterised by the behaviour of the functions at the singularity occurring at the origin.
The self-adjoint extensions were classified in  \cite{AT98} and \cite{DS98} for Schr\"{o}dinger and Pauli operators with the AB field, respectively.
From the recent literature studying such singular interaction let us mention \cite{BCF23, Per06} for results related to Pauli and Dirac operators,
\cite{PR11, CF23} studying Schr\"{o}dinger operators and \cite{DG21, DF23} investigating Bessel operators that include Schr\"{o}dinger operators with the AB field.

\bigskip

Let us sketch the central points of the proofs of our main results Thms.~\ref{thm:bdd} and~\ref{thm:unbdd}. The first steps rely on the idea used in the original Aharonov--Casher paper. In particular the equation for the zero modes 
decouples and we can analyse each component separately. Each of the components then factorizes as a product of an (anti)analytic function $g$ and an exponential whose argument depends on the magnetic scalar potential.
While AC consider the simply connected manifold $\Rb^2$, 
where the function $g$ has a Taylor series, in our case $g$ has only the Laurent series on a neighbourhood of each of the holes. To achieve the starting point of Aharonov and Casher we use the APS boundary condition to 
extend $g$ (anti)analytically to the interior of the holes. 
The main difficulty here is to find a suitable way to compare the boundary values to the boundary condition. This is a local analysis and this step requires that the boundaries of our holes are indeed circular. In the unbounded case of $\Rb^2$ with holes we then 
complete the analysis by cutting off the Taylor series of $g$. For that we use the condition that the zero modes need to be in the domain of the operator and therefore have to be square integrable at infinity. This is the same mechanism as in \cite{AC79}.
For the case of a disc the eigenfunctions have to again satisfy the APS boundary condition, which provides us with the cut off on the Taylor series of $g$.
The highest possible power in the series determines the number of the
zero modes. Due to the different source of the constraint on this power
we arrive at different results. In Rem.~\ref{rem:zero_modes_result_form} we however show that for those values of fluxes, where the results yield a different number of zero modes, the extra zero mode on the disc is not square integrable at infinity when considered on a disc of a growing radius.

\bigskip	 

Finally, we briefly outline the content of this paper.
In this introduction we give the definition of the Dirac operator on 
	an orientable two dimensional Riemannian manifold, 
	and the APS boundary condition.
	We further discuss the magnetic field. 
	Introducing our particular setting we find an explicit form of
	the APS boundary condition and establish the gauge invariance of	
	 the problem in Lem.~\ref{le:gauge_invariance}.
	 Using Lem.~\ref{le:gauge_invariance} we can without loss of generality 
	 study only fluxes mod $2\pi$ inside each hole.
	We refer to this as ``gauging away'' integer-multiples of $2\pi$ 
	of the flux.

	In Sec.~\ref{chap:main_thms} we state and prove the main results.
	To briefly summarise, we obtain the same result as Aharonov and Casher in the case of the Dirac operator on a plane with holes. On a disc with holes our statement is in accordance with the index theorem.
	
	We extend the Aharonov--Casher theorem to a sphere with holes in 
	Sec.~\ref{sec:sphere}.
	Despite this being a direct consequence of our result on a disc with holes, 
	due to the fact that the two cases are related by stereographic projection,
	we first need some theoretical preparation in form of treating the Dirac operator with APS boundary condition in a conformal metric.
	The proof also requires analysis of the spinors under the change of 
	coordinates by the M\"{o}bius transform which we discuss in Appx.~\ref{ap:mobius_transform}.
	
	In Sec.~\ref{sec:index} we use the generalized index formula by 
	Grubb \cite{Gru92} and Gilkey \cite{Gi93}
	of the index theorem on manifolds with boundary,
	to compute the index of the magnetic Dirac operator
	and compare it to our result on the disc region. 
	Let us remark, that the original result in \cite{APS1} 
	cannot be applied directly 
	since it was restricted to manifolds 
	that have a product structure near the boundary.

	Let us mention that
	this work is based on the author's PhD thesis~\cite{Thesis}.

\bigskip

\noindent\textbf{Acknowledgment:}
First and foremost I am grateful to Jan Philip Solovej for fruitful meetings during (and after)
my PhD program, when this work was done.
Further I would like to thank Joshua Hunt, Anna Sisak, Jakub L\"{o}wit, 
B\l{}a\.{z}ej Ruba, Volodymir Riabov, Lukas Schimmer and Georgios Koutentakis for valuable discussions.
Many thanks belong to Rafael Benguria for hosting my visit, during which some of the work has been done.
I am also grateful to Marina Prokhorova who first initiated the discussion of this project topic
and to Annemarie Luger for her valuable comments during my PhD defence
and in particular pointing out the qualitative difference in our two main results.
I would like to acknowledge 
support for research on this paper from VILLUM FONDEN through the QMATH
Centre of Excellence grant. nr. $10059$.
This project also received funding from the European Union’s Horizon 2020 research and innovation
programme under the Marie Sk\l{}odowska-Curie grant agreement No~$101034413$.

I am grateful to the two reviewers for reading carefully my manuscript and pointing out several issues contributing thus significantly to the readability and clarity of this paper.

\subsection{Dirac operator and the APS boundary condition}
\label{sec:geometry}
	Let $M$ be a two-dimensional oriented Riemannian manifold  
	with compact boundary $\p M$ and metric $g$. 
	Let further $E$ be a two-dimensional complex vector bundle equipped 
	with an inner product $(\cdot, \cdot)_E$
	on the fibres of $E$. 
	Denote by $\mathrm{End}(E)$ the bundle of endomorphisms of the bundle $E$.
	If there is a vector bundle map
	\begin{align*}
 		\sigma: T ^{\ast}M  \rightarrow \mathrm{End}(E)
	\end{align*}
	which is Hermitian, \ie
	 $\sigma(\zeta) = \sigma(\zeta)^{\ast}$ 
	for all $\zeta \in T ^{\ast}M$,
	and satisfies the Clifford relations
	\begin{align}
	\label{eq:clifford_rel}
		\sigma(\zeta) \sigma(\mu) + \sigma(\mu) \sigma(\zeta)= 2g(\zeta, \mu) 
		\qquad \text{ for any }
		\zeta, \mu \in T_p ^{\ast}M 
		 \text{ at all } p\in M \,,
	\end{align}
	we call $E$ a \demph{Spin$^c$ spinor bundle}\footnote{
		One might rather call $E$ a bundle 
		of irreducible complex Clifford modules. These two 
		concepts were, however, shown to be identical, see
		\cite[Sec.~3]{Chen17}.
		}
	over $M$.
	The mapping $\sigma$ is called the Clifford multiplication. 
	 The Clifford multiplication further extends to a unique mapping 
	 from the bundle of Clifford algebras $Cl(T^{\ast}M)$. 
	 That is the quotient $\otimes T^{\ast} M / I_g $ 
	 of the bundle of tensor algebras 
	 $\otimes T^{\ast}M \coloneqq \oplus_{k\geq 0} (T^{\ast}M)^{\otimes k}$
	 by the bundle of ideals $I_g$ generated by 
	 $\{\zeta \otimes \zeta -2 g(\zeta, \zeta) \mid 
	  	\zeta \in T^{\ast}M  \}$.
	In even dimensions we call the 
	$\ii^{3\dim(M)/2}$-multiple of the Clifford multiplication 
	by the volume form (which is an involution) 
	the \textbf{chirality operator}
	and refer to its eigenvectors with eigenvalue 
	$+1$ (or $-1$) as \textbf{spin up} (or \textbf{down})
	vectors. 
	Locally we can always choose the representation of
	$Cl(\Rb^2)$ by the constant Pauli matrices
	\begin{align} \label{eq:Pauli_matrices}
		\sigma^1 = \sigma(\dd x) =
			\begin{pmatrix}
			 	0 	& 	1 \\
			 	1	& 	0
			\end{pmatrix}	
		\text { and }
		\quad
		\sigma^2 = \sigma(\dd y) =
			\begin{pmatrix}
			 	0 	& 	-\ii \\
			 	\ii	& 	0
			\end{pmatrix} \,.		
	\end{align}
	 Note that in $\Rb^2$ with the standard metric the chirality operator corresponds to the 
	third Pauli matrix $\sigma_3= \mathrm{diag} (1, -1)$.
	In what follows we use the standard notation 
	$\Gamma(M, E)$
	for smooth sections of the bundle $E$
	and $L^2(M, E)$ for the square integrable sections w.r.t. the volume element
	generated by the Riemannian metric on $M$.
	Let us next equip $E$ with a connection $\nabla$. We call
	$\nabla$ a \demph{Spin$^c$ connection} if
	it is
	\begin{enumerate}
	 \item metric:
	 \begin{align*}
	 	X (\zeta, \mu)_E = (\nabla_{X} \zeta, \mu)_E + (\zeta, \nabla_{X} \mu)_E\,,
	 \end{align*}
	 for any sections $\zeta, \mu \in \Gamma(M, E)$ and any vector field $X \in TM$, and,
	 \item compatible with the Clifford multiplication ${\sigma}$:
	 \begin{align*}
	  	[\nabla_X, {\sigma}(\mu)] = {\sigma}(\nabla_X^{LC}\mu) \,,
	 \end{align*}
	 for all vector fields $X$ and one-forms $\mu$ on $M$. 
	 Here, $\nabla^{LC}$ is the Levi-Civita connection 
	 on the cotangent space $T^{\ast}M$ of $M$.
	\end{enumerate}
	\begin{definition}\label{def:Dirac_op}
		Let $E$ be a Spin$^c$ spinor bundle over $M$
		with Clifford multiplication $\sigma$ and
		a $Spin^c$ connection $\nabla$.
		The Dirac operator $D:\Gamma(M, E) \rightarrow \Gamma(M, E)$ is the following composition
		\begin{align*}
			D = -\ii \sum_{j\leq 2} {\sigma}(e^j) \nabla_{e_j} \,,
		\end{align*}
		where $(e_j)_{j \in \{1,2\}}$ is an orthonormal basis on $TM$ and $(e^j)_{j\leq 2}$ is
		the corresponding dual basis on $T^{\ast}M$.
	\end{definition}
	Note that the definition is independent of a particular choice of $(e_j)_{j\leq 2}$, 
	so $D$ is well defined globally. 
	We remark, that one can use a Spin$^c$ connection on a Spin$^c$ 
	bundle to form a Dirac operator with magnetic field 
	(see further Sec.~\ref{ssec:mag_field}).
	The Dirac operator is a first order operator which is elliptic, symmetric  and
	 whose principal symbol is the Clifford multiplication.
	Furthermore it can be extended as a closed linear map
	to the \demph{maximal domain} of $D$ 
	 \begin{align}\label{eq:dom_max}
	 	\dom(D ^{\max}) \coloneqq \{ u \in L^2(M, E) \mid D u \in  L^2(M, E)  \} \,.
	\end{align}	
	To introduce the Atiyah--Patodi--Singer (APS) boundary condition
	we are following the formalism for elliptic boundary 
	conditions used 
	in~\cite{BB, BB12}. 
	We, however, diverted with the convention for the Clifford
	multiplication which in
	the cited papers is considered to be anti-hermitian and satisfies 
	the Clifford relations~\eqref{eq:clifford_rel}
	with an extra minus sign on the right-hand side.
	\begin{notation} \label{notation}
		Let $\nu^{\sharp}$ be the normalized
		inner normal vector field on $\p M$.
		We will denote by $\nu \in T ^{\ast}M$ the local co-vector field 
		on the boundary $\p M$ dual to $\nu^{\sharp}$.
		The local space of co-vectors tangent to the boundary is defined by
		\begin{align*}
		 	T ^{\ast}\p M \coloneqq \{\xi \in T ^{\ast }M \mid g(\xi, \nu) = 0\} \,.
		\end{align*}
		 We further write $\xi^{\sharp}$ for the normalized tangent vector which is dual to $\xi \in T ^{\ast}\p M$.
	\end{notation}	
	With this notation we can locally write the Dirac operator 
	in the neighbourhood of the boundary
	as
	\begin{align}\label{eq:boundary_form}
 		D = -\ii \sigma(\nu) (\nabla_{\nu^{\sharp}} + A_0)
 		\quad \text{with} \quad
 		 A_0 = \sigma(\nu) \sigma(\xi) \nabla_{\xi^{\sharp}} 	\,,
	\end{align}
	where we used that by the Clifford relations $\sigma(\nu)^2$ is the identity on fibres of $E$ (restricted to $\p M$).
	As in the Appx.~2 of \cite {BB} we then define the canonical boundary operator 
	which anti-commutes with $\sigma(\nu)$.
	\begin{definition} \label{def:canonicalBO}
		Let $E$ be a Spin$^c$ spinor bundle over $M$
		with Clifford multiplication $\sigma$ and
		a $Spin^c$ connection $\nabla$
		and let $D$ be a Dirac operator on $E$. 
		We define the \demph{canonical boundary operator adapted to $D$} by
		\begin{align*}
			A \coloneqq 
				 \frac{1}{2}(A_0-\sigma(\xi)\nabla_{\xi^{\sharp}} \sigma(\nu))
				= A_0 -\frac{\kappa}{2} \,,
		\end{align*}
		where $\kappa$ is the eigenvalue of the shape 
		operator of the 
		boundary w.r.t. the normal field $\nu$, \ie
		$\nabla^{LC}_{\xi^{\sharp}}\nu = \kappa \xi$.
		In fact $\kappa$ is the principal curvature 
		of the boundary.
	\end{definition} 
Note, that $A$ was chosen so that the anti-commutator 
$\{A, \sigma(\nu)\}$ vanishes. 
In our two-dimensional case it is also not difficult to check by a direct computation that $A_0$ is symmetric.
For a general dimension this is shown in \cite[Appx.~1]{BB}. By definition the canonical boundary operator  $A$ is thus also symmetric.  It follows, that it is essentially self-adjoint on $L^2(\p M; \Cb^2)$, since $\p M$ is a compact manifold.
The importance of boundary operators is that one can use them for a construction
of elliptic boundary conditions (see \cite[Defs.~1.9,~1.10, and Thm.~1.12]{BB12}) which give rise to domains that are subsets of
 $H^1_{loc}(M, E) = \{u \in \ L^2_{loc}(M, E) \mid \nabla u \in L^2_{loc}(M, E) \}$.
 Here $L^2_{loc}(M, E)$ denotes sections of $E$ that are square integrable 
 over each compact subset $K \subset M$ and, in particular, we may have
 $K\cap \p M \neq \emptyset$.
  
	  Since $A$ is a self-adjoint elliptic operator on the compact manifold $\p M$, it has purely discrete spectrum.
	Let us denote by $\{ v_k \mid k\in \Zb\setminus \{0\} \}$ 	
	a set of orthonormal
	eigenvectors of $A$ 
	corresponding to eigenvalues $\lambda_k\neq 0$. 
	We order these eigenvalues as
	$\ldots\leq\lambda_{-k}\leq \lambda_{-k+1}\ldots \leq \lambda_{-1}<0<\lambda_1\leq \ldots<\lambda_k\leq \lambda_{k+1}\leq \ldots$.
	Let us further assume that there is the decomposition
	into two mutually orthogonal spaces
	 $\ker A = \mathrm{span }\{v_0\} \oplus
	 	\mathrm{span }\{ \sigma(\nu) v_0\}$.
	In general $v_0$ could be a set of vectors but in our case it is only one vector.
	We define the \demph{APS (Atiyah--Patodi--Singer) boundary condition} as the following closure
	of a subset of smooth sections on the boundary 
	\begin{align}\label{eq:APS_general}
	 	BC _{APS}  \coloneqq \overline{\mathrm{span} (\{v_k\}_{\lambda_k<0} \cup v_0 )} \,.
	\end{align}
	We point out that $v_0$ and $\sigma(\nu)v_0$ are exchangeable
	and that we are making a choice here.
	The closure in~\eqref{eq:APS_general} is w.r.t. the norm
	\begin{align}
	\label{eq:H_check}
	 	\bigg\|\sum_{k \in \Zb} c_k v_k\bigg\|^2 _{\check{H}(A)} 
	 		\coloneqq \sum_{\lambda_k < 0} |c_k|^2(1+\lambda_k^2) ^{1/2} 
	 			+ \sum_{\lambda_k \geq 0} |c_k|^2(1+\lambda_k^2) ^{-1/2}
	 			\,,
	 			\quad
	 	c_k \in \Cb \,.
	\end{align}
	Further, $\check{H}(A)$ denotes the closure of $C^{\infty}(\p M, \Cb^2)$
	in this norm.
	We call the realisation $D^{\APS}$ of $D$ on the domain
	 \begin{align} \label{eq:APS_domain}
		\dom (D^{\APS}) = \{ u \in \dom(D^{\max}) \mid \gamma_0 u \in BC _{APS}\} \,,
	\end{align}
	the \demph{Dirac operator with APS boundary condition}.
	The trace map $\gamma_0 u = u\big|_{\p M}$ is well defined for 	
	$u\in C_0^{\infty}(M, \Cb^2)$,
	in particular $\supp u \cap \p M$ can be non-empty,
	and by \cite[Thm.~1.7.(2)]{BB12} it extends to a bounded linear map
	\begin{align}\label{eq:trace_map}
	 	\gamma_0: \dom(D^{\max})\rightarrow \check{H}(A) \,.
	\end{align}
	Thm.~4.12. in~\cite{BB} tells us that \eqref{eq:APS_domain} is a self-adjoint realisation.
	Recalling that $\sigma(\nu)$ and $A$ anti-commute, the self-adjointness can be also seen directly from the Green's formula (\cf
	\cite[Prop.~2.1]{BB})
	\begin{align*}
		\int_{M}(D\psi, \phi)_E = \int_{M}(\psi, D^{\ast} \phi)_E 
		 - \int_{\p M}(-\ii \sigma(\nu) \psi, \phi)_E \,.
	\end{align*}
	Here $\psi, \phi$ are compactly supported functions on $M$. In particular their support and the boundary $\p M$ can have a non-empty intersection. By $D^{\ast}$ we denoted the adjoint of $D$.
	 Let us remark that there are also other choices of more general 
	 APS boundary conditions (see \cite[Example~1.16(b)]{BB12}). 	
	If $M$ is compact then Thm.~5.3 in \cite{BB} implies that the Dirac operator $D^{\APS}$
	with the APS boundary condition is Fredholm.
	We recall that Fredholm operator is an operator with closed range and a finite dimensional kernel and cokernel. 
	
	In this work we are also interested in $M$ being a plane with holes. 
	In this case zero is an eigenvalue of finite geometrical multiplicity embedded in the essential spectrum (see Thm.~\ref{thm:unbdd} and Cor.~\ref{cor:zero_in_ess_spec} bellow)
	and therefore $D^{\APS}$ does not have a closed range (\cf \cite[Thm.~5.2]{Kato}). Thus it is not Fredholm.
	Since we will consider manifolds with several components of boundary we introduce the notation $BC _{APS}\mid_{\p \Omega_j}$
	for the $APS$ boundary condition on the component $\p \Omega_j$ of the boundary.

\subsection{Magnetic field and minimal coupling}	\label{ssec:mag_field}
	Let us consider the connection  $\nabla = \dd -\ii \alpha$ on a
	the trivial bundle $E$ over $\mathscr{\Cb}$ with fibre $\Cb^2$.
	The term $-\ii \alpha$ is called the local connection one-form and it satisfies
	$-\ii \alpha(Y) \in \ii \Rb$ for all vector fields $Y$ on $\Cb$.
	Writing $\alpha = \frac{1}{2} (a \dd \bar{z} + \bar{a} \dd z)$, $a \in C^{\infty}(\Cb)$ 
	and
	using the standard notation $\p_z = \frac{1}{2} (\p_x-\ii \p_y)$, $\p_{\bar{z}} = \frac{1}{2} (\p_x+\ii \p_y)$ we obtain the Dirac operator (\cf the representation of the Clifford multiplication 
	\eqref{eq:Pauli_matrices})
	\begin{align*}
 		D_a = 
 			-2\ii 
 		\begin{pmatrix}
 			0 			& 	\p_z	\\
 			\p _{\bar{z}} 	& 	0
		\end{pmatrix}
		- 
		\begin{pmatrix}
 			0 	& 	\overline{a} 	\\
 			a 	& 	0
		\end{pmatrix} \,,
	\end{align*}
	known from physics to be (by the principle of correspondence 
	and the minimal coupling) the Hamiltonian of a relativistic charged mass-less particle in a magnetic field of vector potential $\alpha$.
	The field strength is the closed two-form $\beta=\dd \alpha$ and
	two different connection one-forms $\alpha ^{(1)}, \alpha ^{(2)}$ correspond to the same magnetic field $\beta$
	if they differ by an exact form. This ambiguity is the well known gauge invariance.
	To put this in the context of the vector formalism we can write $a = a_x + \ii a_y$, for some $a_x, a_y \in C^{\infty}(\Rb^2)$.
	Then defining
	\footnote{$g^{jk}$ here refers to the components of the metric $g$ in the coordinate basis $(\dd x, \dd y)$ of one-forms} 
	$a^j = g^{jk} \, a_k$
	 for $j,k \in \{x,y\}$,
	the vector potential $\vec{a} = (a^x, a^y)$
	corresponds to the magnetic field $\vec{B} = (0, 0, B)$ such that
	$\mathrm{curl} (a^x, a^y, 0) = (0, 0, B)$ with $\beta = B(z) \frac{\ii}{2} \dd z \wedge \dd \bar{z}$.
	We will introduce the Aharonov--Casher gauge  
	\begin{align}\label{eq:h_relation}
	 	\partial_{z} h(z)  = -\frac{\ii \overline{a} }{2} \,,
	\end{align}
	using the scalar potential $h$ satisfying $-\Delta h = B $ on $\Cb$.
	For $B$ decaying sufficiently fast at infinity 
	(not necessarily smooth) we can write the solution of this Poisson equation
	 \begin{equation}\label{eq:potential}
	 	h(z) = - \frac{1}{2\pi } \int_{\Cb}\log |z- z'| B(z') \frac{\ii}{2}\dd z' \wedge \overline{\dd z'} \,.
	\end{equation}
	Notice that this gauge is automatically divergence free 
	\begin{align*}
 		\p_x a_x + \p_y a_y
 			= \Re (2\p _{\bar{z}} \bar{a})
	 		= \Re (4\ii \p _{\bar{z}}\p _{z} h) = \Re (-\ii B) = 0 \,.
	\end{align*}
	 
	Another quantity that describes the magnetic field is called the \demph{magnetic flux}
	\begin{align*}
	 	\Phi\coloneqq \int _{\Cb} B \frac{\ii}{2} \dd z \wedge \dd \bar{z} \,.
	\end{align*}
	
	\begin{remark}\label{rem:a_is_bounded}
		Let us now consider a smooth magnetic field $B$ with compact support.
		By elliptic regularity (see \eg \cite[Sec.~6.3, Thm. 3]{Evans}) the potential $h$ 
		is then also a smooth function.
		Note that the Poisson equation $-\Delta h = B$ determines $h$
		 up to an addition of a harmonic function
		 and our particular choice corresponds to the unique gauge choice via the relation~\eqref{eq:h_relation}.
		 It yields a divergence free vector potential $a(z)$ that is bounded at infinity. 
		 We will refer to the 
		 choice~\eqref{eq:h_relation}, \eqref{eq:potential} of 
		 $a(z)$ as the \demph{Aharonov--Casher gauge}.
		
		 To see the boundedness, let $R'>0$ be such, that $R'>2|z'|$ for all $z' \in \supp B$ and 
		 taking $|z|>R'>2|z'|$ we can use the bound 
		 \begin{align*}
			 	\bigg| \frac{B(z')}{z-z'} \bigg| \leq \frac{2}{R'} |B(z')| 
			 		\in  L^1(\Cb) \,, \quad z'\in \supp B \,,
			\end{align*}
		to apply the dominated convergence theorem and obtain
		\begin{align*}
			 	|\p_z h(z)| 
			 		\leq \frac{const}{|z|} \int _{\Cb} \frac{|B(z')|}{1-\frac{|z'|}{|z|} } \frac{\ii}{2} \dd z'  \wedge \overline{\dd z'}
			 		\leq \frac{const}{|z|} \int _{\Cb} |B(z')| \ii \dd z'  \wedge \overline{\dd z'}
			 		\leq \frac{const}{|z|} \,,
			\end{align*} 
		for $|z|$ large.
		
		From \eqref{eq:potential} one also deduces the asymptotic behaviour of the scalar potential
		\begin{align}\label{eq:asymptot_h}
			h(z) = -\frac{\Phi}{2\pi} \log |z| + \mathcal{O}(|z| ^{-1}) \,,
		\end{align}
		as $|z|$ tends to infinity. 
		Moreover in the case of a spherically symmetric $B$ there is no error term
		and for $z$ outside of support of $B$
		\begin{align}\label{eq:h_spherical_field}
		 	h(z) = -\frac{\Phi}{2\pi} \log|z|  \,, 
		\end{align}
		by the Newton's law.
	 \end{remark}

	\begin{remark}\label{rem:vec_pot_in_polar}
	 	In the flat space $\Rb^2$ we clearly have $a_x =a^x$ and
		$a_y = a^y$. Notice that if we consider the polar coordinates 
		$(r, \phi) = (\sqrt{x^2 + y^2}, \arctan \frac{y}{x} )$
		and write 
		$\alpha = (a_r, a _{\phi})$ for the components in the normalized
		basis $(\dd r, r \dd \phi)$ 
		and $\vec{a} = (a^r, a ^{\phi})$ for the components in the dual basis
		$\left( \p_r, \frac{\p _{\phi}}{r}\right)$ we also have $a_r = a^r$ 
		and $a _{\phi} = a ^{\phi}$.
	\end{remark}
	
\subsection{Problem setup}
\label{sec:problem_setting}
We start by establishing some notation:
	\begin{itemize}
		\item Let $\mathscr{M}$ be either 
				the complex plane or a disc $\Omega_{out} \subset \Cb$ with centre 
				at the origin with radius $R_{out}$.
	 	\item 	$\Omega_j \subset \mathscr{M}$, $j \in \{1,2,\ldots, N\}$
	 			refers to a ball  with centre at $w_j\in \mathscr{M}$ 
	 			and radius $R_j>0$.
	 	\item $M = \mathscr{M} \setminus \cup_{k\leq N} \Omega_k$, $N \in \Nb$ is our
	 			two-dimensional manifold of interest.
	 	\item $(r_j, \phi_j)$ denote the polar coordinates at $w_j \in \Omega_j$.
	 	\item $B_j$, $j\leq N$ denotes the magnetic field with support inside $\Omega_j$,
	 			while $B_0 \in C_0^{\infty}(M^{\circ})$.
	\end{itemize}
	Complementing the above notation for magnetic field on $\mathscr{M}$ we denote
	\begin{align} \label{eq:tot_mag_field}
		 B = B_{sing} + B_0 \,,
	\end{align}
	where $\supp B_{sing} \subset \cup_{k\leq N} \Omega_k$.
	Later in Lem.~\ref{le:gauge_invariance} we will show that without loss of generality we may assume
	$B_{sing} = \sum_{k\leq N}\Phi'_k \delta_{w_k}$,
	where $\Phi'_k \in [-\pi, \pi)$ \footnote{
		we can choose any interval of length $2\pi$, 
		but this choice is the most convenient one for the purposes of our analysis
	} 
	differs by an integer multiple of $2\pi$ from the flux of $B$
	through the $k$-th hole
	\begin{align*}
	 	\Phi_k \coloneqq \int_{\Omega_k} B(z) \frac{\ii}{2}\dd z\wedge \dd \bar{z} \,.
	\end{align*} 
	We refer to $\Phi'_k$ as a normalized flux through the hole $\Omega_k$.
	The total flux is then the sum 
	\begin{align*}
	 	\Phi \coloneqq \Phi_0 + \sum_{k\leq N}\Phi'_k \,,
	\end{align*}	 
	where $\Phi_0 = \int_{M} B_0(z) \frac{\ii}{2}\dd z\wedge \dd \bar{z}$
	is the flux through the bulk of $M$.
	
\subsubsection*{The Dirac operator and the APS boundary condition explicitly}
	For finding more concrete form of the APS boundary condition 
	for the above described setting
	we will make use of the Dirac operator expressed in polar coordinates $(r,\phi)$ 
	\begin{align} \label{eq:D_in_polar}
 		D_a 
 		= -\ii	
 		\begin{pmatrix}
			0		&	\ee^{-\ii \phi}(\partial_r- \ii \frac{\partial_{\phi}}{r}) 	\\				
			\ee^{\ii \phi}(\partial_r+ \ii \frac{\partial_{\phi}}{r})		& 	0
		\end{pmatrix}
		- 
		\begin{pmatrix}
 			0		& 	\ee^{-\ii \phi} (a_r - \ii  a_{\phi}) 	\\
 			\ee^{\ii \phi} (a_r + \ii a_{\phi}) 	& 		0
		\end{pmatrix}	\,,
	\end{align}
	where we use notation from Remark~\ref{rem:vec_pot_in_polar}.

	To gain some intuition for the abstract setting, we first work out an example of finding the boundary condition for the case of one hole.

\begin{small}
\begin{example} \label{ex:one_hole}
		Consider the manifold  $M=\Cb\setminus \Omega$ with $\Omega$ being a ball of radius $R$ centred at the origin. We assume there is no magnetic field in the bulk (\ie $B_0 = 0$) and put magnetic field formed by one Aharonov--Bohm flux $B = \Phi \delta$ inside the hole. Later, 
		in Le.~\ref{le:gauge_invariance} we show that without loss of generality we can always consider that the magnetic field inside a hole is of this form and, moreover, that $\Phi \in [-\pi, \pi)$.
		The key simplifying point is that for this field we can choose the gauge so that $a_r = 0$ and $a_{\phi} = \frac{\Phi}{2\pi r}$.
	Note, that in this setting the inward normal one-form on the boundary is simply $\nu = \dd r$.
	 In accordance we will denote $\sigma(\nu) = \sigma^r$.
	The Clifford connection along the radial field $\p_r$ is given by $\nabla_{\p_r} = \p_r - \ii a_r$.
	To find the boundary operator $A_0$ we compare  \eqref{eq:boundary_form}
	to the expression \eqref{eq:D_in_polar} for $D$ in polar coordinates near the boundary 
	$\p \Omega$ and conclude 
	\begin{equation*}
	 	\sigma^r =  \begin{pmatrix}
	 					0				& 	\ee^{-\ii \phi}	\\
	 					\ee^{\ii \phi} 	& 	0
					\end{pmatrix}
				= (\sigma^r)^{\ast} = (\sigma^r)^{-1} \,.
	\end{equation*}
	With \eqref{eq:boundary_form} this further yields
	$A_0=  	 \sigma^3 \left( \frac{ \ii\p_{\phi}}{R} + a_{\phi}\right)$, 
	 		where $\sigma^3 = \mathrm{diag} (1,-1)$ 
	 		is the third Pauli matrix.
	Finally, using that the principal curvature of the circular boundary $\p \Omega$ is $\frac{1}{R}$ we obtain by Def.~\ref{def:canonicalBO} the canonical boundary operator on $\p \Omega$
\begin{align*}
		A = \sigma^3 R^{-1}\left(\ii \p_{\phi} + \Phi/2\pi \right) - (2R)^{-1} \,.
	\end{align*}
	A simple analysis of the eigenvalueproblem 
	$Au = \lambda u$ then reveals the eigenfunctions
	\begin{align*}
		\begin{pmatrix}
		 	\ee^{\ii k \phi} \\
		 	0
		\end{pmatrix}
		\quad \text{with eigenvalue } 
		\lambda =R^{-1} \left(-k - \frac{\Phi}{2\pi}-\frac{1}{2} \right) \quad \text{ and } \quad
		\begin{pmatrix}
			0	\\
		 	\ee^{\ii k \phi} 
		\end{pmatrix}
		\quad \text{with eigenvalue } 
		\lambda =R^{-1} \left( k - \frac{\Phi}{2\pi}-\frac{1}{2} \right ) \,.
	\end{align*}
	
	The APS boundary condition is therefore given by the closure of the set
	\begin{align} \label{eq:BC_one_hole}
	BC_{APS}\mid_{\p\Omega}
	=
	\overline{
		\mathrm{span} 
			 	\Bigg\{ 
			 	\left[
	 	 			\begin{pmatrix}
 	 			 	\ee^{\ii k \phi}	\\
	 	 			 	0
	 	 			\end{pmatrix}
	 	 		 \right]_{k > \frac{\Phi}{2\pi} -\frac{1}{2}}, 
	 	 	 \left[
	 	 			\begin{pmatrix}
	 	 			 	0 	\\
	 	 			 	\ee^{\ii k \phi}
	 	 			\end{pmatrix}
		 	 	 \right]_{k \leq \frac{\Phi}{2\pi} +\frac{1}{2}}
		 	 	 \Bigg\}
		 }
	\end{align}
	in the $\check{H}$ norm \eqref{eq:H_check}.
	For completion we write down the domain of the operator 
	$D_a^{\APS}$ 
	in this setting,	
	acting as \eqref{eq:D_in_polar} on

	\begin{align*}
		\dom(D_a^{\APS}) = \{\psi \in \dom(D_a^{\max})\mid
					\gamma_0 \psi \in BC_{APS}\mid_{\p\Omega} \} \,.
	\end{align*}

	Notice that the sign of the eigenvalues of the boundary operator is independent of the radius of the hole and the boundary condition is invariant under the scaling. However, it does not directly make sense to consider $R \rightarrow 0 $ which might otherwise be interesting to investigate in view of the many extensions of the Dirac operator with Aharonov--Bohm field (\cf \cite{Per06, BCF23}).	
\end{example}
\end{small}
	
	Now, let us fix an index $j\leq N$ and work out the boundary 
	condition on the boundary of a $j$-th hole in the general setting. 
	We choose the polar coordinates $(r_j, \phi_j)$ around the centre of 
	the hole $\Omega_j$. Then noting that in our case 
	$\nabla_{\nu^{\sharp}} = \p_{r_j} - \ii a_{r_j}$ we obtain
	by~\eqref{eq:boundary_form} 
	$A_0 \mid _{\p \Omega_j} = \sigma^3 \left(\frac{\ii \p _{\phi_j}}{R_j}+ a_{\phi_j}\right)$
	and consequently,
	Def.~\ref{def:canonicalBO} yields
	\begin{align} \label{eq:can_bd_op_j}
 		A \mid _{\p\Omega_j}= \sigma^3\left(\frac{\ii \p _{\phi_j}}{R_j}+ a_{\phi_j} \right) -\frac{1}{2 R_j} \,.
	\end{align}
	Solving the eigenvalue problem for  $A\mid _{\p\Omega_j}$ 
	on the circle $\p \Omega_j$
	we find 
	\begin{align} \label{eq:eigenproblem_bdry_operator}
	 	A\mid _{\p \Omega_j} 
	 	\begin{pmatrix}
	 	 	\psi^j_{\ell} \\
	 	 	0
	 	\end{pmatrix}
	 	=
	 	\lambda^j_{\ell}
	 	\begin{pmatrix}
	 	 	\psi^j_{\ell} \\
	 	 	0
	 	\end{pmatrix}
	 	\quad \text{and }
	 	A\mid _{\p\Omega_j} 
	 	\begin{pmatrix}
	 	 	0\\
	 	 	\psi^j_{\ell}
	 	\end{pmatrix}
	 	=
	 	-\lambda^j_{\ell-1}
	 	\begin{pmatrix}
	 	 	0\\
	 	 	\psi^j_{\ell}
	 	\end{pmatrix}
	 	\quad \text{ with }
	 	\begin{cases}
	 	 	R_j\lambda^j_{\ell} = \Phi_j/2\pi -1/2 - \ell 	\\
			\psi^j_{\ell}  = \ee^{\ii \phi_j \ell}
							\ee^{\ii \int _{\gamma_j} \vec{a}\cdot \dd \vec{s} - \ii \frac{\Phi_j}{2\pi} \phi_j }
	 	\end{cases} \,,
	\end{align}
	where $\ell$ is an integer and the path $\gamma_j \subset \p \Omega_j$  connects the points $(R_j, 0)$
	and $(R_j, \phi_j) \in \p\Omega_j$.
	This further
	leads to the APS boundary condition~\eqref{eq:APS_general} on our chosen 
	component of the boundary
	\begin{align} \label{eq:BC_hole_j} 
	 	BC _{APS}\mid_{\p \Omega_j}	
	 		= 
	 		\overline{\mathrm{span} 
			 	\Bigg\{ 
			 	\left[
	 	 			\begin{pmatrix}
 	 			 	\psi_{\ell}^j 	\\
	 	 			 	0
	 	 			\end{pmatrix}
	 	 		 \right]_{\ell> \frac{\Phi_j}{2\pi} -\frac{1}{2}}, 
	 	 	 \left[
	 	 			\begin{pmatrix}
	 	 			 	0 	\\
	 	 			 	\psi_{\ell}^j
	 	 			\end{pmatrix}
		 	 	 \right]_{\ell \leq \frac{\Phi_j}{2\pi} +\frac{1}{2}}
		 	 	 \Bigg\} 
		 	 	 } \,,
	\end{align}
	where the closure is in the norm of $\check{H}(A|_{\p \Omega_j})$ 
	(see~\eqref{eq:H_check}).
	Denoting $(r, \phi)$ the polar coordinates at the origin
	we notice that the inner normal vector on the boundary $\p \Omega_{out}$ corresponds to $-\p_r$.
	Taking this into account means that formally 
	$A \mid _{\p \Omega _{out}}$ looks like~\eqref{eq:can_bd_op_j} with a minus sign,
	$A \mid _{\p \Omega _{out}} = -\sigma^3 \left(\frac{\ii \p _{\phi}}{R _{out}}+ a _{\phi}\right) +\frac{1}{2 R_{out}}$.
	Therefore we infer the solution of the eigenvalue problem 
	immediately from~\eqref{eq:eigenproblem_bdry_operator}
	\begin{align} \label{eq:eigenproblem_bdry_operator_out}
	 	A\mid _{\p \Omega_{out}} 
	 	\begin{pmatrix}
	 	 	\psi_{\ell} \\
	 	 	0
	 	\end{pmatrix}
	 	=
	 	-\lambda_{\ell}
	 	\begin{pmatrix}
	 	 	\psi_{\ell} \\
	 	 	0
	 	\end{pmatrix}
	 	\text{ and }
	 	A\mid _{\p\Omega_{out}} 
	 	\begin{pmatrix}
	 	 	0\\
	 	 	\psi_{\ell}
	 	\end{pmatrix}
	 	=
	 	\lambda_{\ell-1}
	 	\begin{pmatrix}
	 	 	0\\
	 	 	\psi_{\ell}
	 	\end{pmatrix}
	 	\quad \text{ with }
	 	\begin{cases}
	 	 	R_{out}\lambda_{\ell} = \Phi/2\pi -1/2 - \ell 	\\
			\psi_{\ell}  = \ee^{\ii \phi \ell}
							\ee^{\ii \int _{\gamma_{out}} \vec{a}\cdot \dd \vec{s} - \ii \frac{\Phi}{2\pi} \phi }
	 	\end{cases} \,,
	\end{align}
		where $\gamma _{out}\subset \p \Omega_{out}$ connects the points $(R _{out}, 0)$ and $(R _{out}, \phi)$.
	The APS boundary condition on the outer component of the boundary 
	thus reads
	\begin{align}
	 \label{eq:BC_hole_N+1}
		BC _{APS}\mid _{\p \Omega _{out}} 
	 	 	 = \overline{\mathrm{span} 
		 	\Bigg\{ \left[
	 	 			\begin{pmatrix}
	 	 			 	\psi_{\ell}\\
	 	 			 	0
	 	 			\end{pmatrix}
	 	 		 \right]_{\ell < \frac{\Phi}{2\pi} -\frac{1}{2}}, 
	 	 	 \left[
	 	 			\begin{pmatrix}
	 	 			 	0 	\\
	 	 			 	\psi_{\ell}
	 	 			\end{pmatrix}	
		 	 	 \right]_{\ell \geq \frac{\Phi}{2\pi} +\frac{1}{2}}	
	 	 	 \Bigg\}  }	\,,
	\end{align}	
	where the closure is in the norm of $\check{H}(A|_{\p \Omega_{out}})$ 
	(see~\eqref{eq:H_check}).

The canonical APS boundary condition is gauge invariant in the sense of the following lemma.
\begin{lemma} \label{le:gauge_invariance}
	Let $D^{\APS}_a$ and $D_{\widetilde{a}}^{\APS}$ be two Dirac operators
	with the APS boundary condition
	 on $M$ corresponding to magnetic fields with fluxes
	$\Phi$ and $\widetilde{\Phi}$ respectively, such that
	\begin{align*}
		\Phi 				=\sum _{j\leq N} \Phi_j + \Phi_0	 \quad 
		\text{ and } \quad
		\widetilde{\Phi}	= \sum _{j\leq N}  \widetilde{\Phi}_j + \Phi_0	\,,
	\end{align*}
	 where
	$\Phi_j$ and $\widetilde{\Phi}_j$ are the fluxes through the hole $\Omega_j$, $j\leq N$
	of $a$ and $\widetilde{a}$ respectively,
	and $\Phi_0$ is the flux of a smooth magnetic field 
	supported inside the interior of $M$.
	If for all $j\leq N$
	\begin{align*}
 		\widetilde{\Phi}_j &= \Phi_j+ m_j 2\pi\,, 
	\end{align*}
	for some $m_j \in \Zb$,
	then $D_a^{\APS}$ and $D_{\widetilde{a}}^{\APS}$ are unitarily equivalent
	\begin{align*}
 		\mathscr{U}^{\ast} D_a^{\APS} \mathscr{U} 
 			= D_{\widetilde{a}}^{\APS} 	\,,
	\end{align*}
	with the unitary operator
 	\begin{align*}
 		\mathscr{U} &: L^2(M, \Cb^2)  \rightarrow L^2(M, \Cb^2) \,,	\\
 		\mathscr{U} &: u \mapsto \exp \bigg[\ii \int_{\gamma} (\vec{a} -\widetilde{\vec{a}}) \dd \vec{s} \bigg] u \,,
	\end{align*}
where $\gamma$ connects a fixed point $z_0 \in M$ and the point $z \in M$.
\end{lemma}

\begin{proof}
 First we notice that $\mathscr{U}$ is independent of a particular choice
  of the path $\gamma$ in its definition
  as for an arbitrary loop $\gamma \subset M$ it is an identity operator
  $\exp \left[\ii \oint_{\gamma} (\vec{a} -\widetilde{\vec{a}}) \dd \vec{s} \right] = 1$.
This we can see immediately from the equalities
\begin{align*}
	\oint_{\gamma} (\vec{a} - \widetilde{\vec{a}}) \dd \vec{s}	
		=\int_{int \, \gamma} B-\widetilde{B}
		= - 2\pi \sum _{\{ j\mid \Omega_j \subset int \, \gamma \}}  m_j	 \,,
\end{align*}  
where in the first equality we used the relation 
$\mathrm{curl}\, \vec{a} = (0,0,B)$
and Stokes theorem.
Notice also that choosing another point $z_1\in M$ as a starting point of $\gamma$ instead of $z_0$
amounts to multiplication by a constant 
$K =  \ee^{- \ii \int_{z_1}^{z_0} (\vec{a} - \widetilde{\vec{a}}) \dd \vec{s} }$.
Since $\overline{K} = K^{-1}$ 
the map $K^{-1}\mathscr{U}$ is also unitary.
An explicit computation of the partial derivatives
\begin{align*}
	\p_z \mathscr{U}
		= \frac{\ii}{2} \overline{(a - \widetilde{a})}(z) \mathscr{U} \,, \quad \text{ and }	
	\quad
	\p_{\bar{z}} \mathscr{U}
		= \frac{\ii}{2} (a - \widetilde{a})(z) \mathscr{U} 
\end{align*}
shows, that the action of the unitarily transformed Dirac operator is indeed the one with the potential $\widetilde{a}$, as
\begin{align*}
	\mathscr{U}^{\ast} D_a \mathscr{U}
		 = D_a - 2\ii
		 	\mathscr{U}^{\ast} 
		 	\begin{pmatrix}
		 	 	0				& \p_z 	\\
		 	 	\p_{\bar{z}}	&	0
		 	\end{pmatrix}
		 	\mathscr{U}
		 = 
		 D_a - \ii
		 	\begin{pmatrix}
		 	 	0					& \ii(\overline{a - \widetilde{a}}) 	\\
		 	 	\ii(a - \widetilde{a}) 	&	0
		 	\end{pmatrix} \,.
\end{align*}
Moreover, taking our observation that $\overline{K} = K^{-1}$ into account we see
that the relation $\mathscr{U}^{\ast} D_a \mathscr{U} = D_{\widetilde{a}}$ 
holds independently of the starting point of $\gamma$.

Finally  we need to check that the boundary condition is preserved by $\mathscr{U}$.
 To do so, we show that the boundary operators are unitarily equivalent 
\begin{align} \label{eq:equiv_rel_BO}
 	A (\widetilde{a}) = \mathscr{U} ^{\ast} A(a) \mathscr{U}\,.
\end{align}
Here $A (a)$ denotes the canonical boundary operator adapted to $D_a$
and $\gamma$ in the definition of $\mathscr{U}$ ends 
at a point on the boundary $z \in \p \Omega$.
From this we see that the restriction to the boundary of a spinor $u\big| _{\p \Omega_j}$ is in the negative spectral subspace of $A(\widetilde{a})$ 
if and only if $(\mathscr{U}u)\big| _{\p \Omega_j}$ is in the negative spectral subspace of $A(a)$.
To see that~\eqref{eq:equiv_rel_BO} holds,
we write $\alpha = a_{\rho} \nu + a_s \xi$ 
with some smooth functions $a_{\rho}, a_s$ on a neighbourhood of the boundary and 
consider $z = \gamma(s) \in \p \Omega_j$.
By path independence we have
 $\mathscr{U} = K \ee ^{\ii \int_0^s (a_s -\widetilde{a}_s) \dd s}$
 with $s$ the arc parameter of $\gamma \subset \p \Omega_j$.
In this notation the boundary operator from Def.~\ref{def:canonicalBO} reads
\begin{align*}
 	A(a)
 		= \sigma(\nu) \sigma(\xi) 
 		\left(\p _s -\ii a_{s} \right) -\kappa/2 \,,
\end{align*}
and we can easily compute its commutator with $\mathscr{U}$ 
using the fundamental theorem of calculus 
\begin{align*}
 	[A(a), \mathscr{U}] 
 		= \sigma(\nu) \sigma(\xi) \p _s (\mathscr{U})	
 		=  \sigma(\nu) \sigma(\xi) \mathscr{U}(z) (\ii a _{s}(z) - \ii \widetilde{a} _{s}(z))  \,,
\end{align*}
which yields
$A(a) \mathscr{U} u
		= \mathscr{U} A(a)  u  + [A(a), \mathscr{U}] u 	
 		= \mathscr{U} A(\widetilde{a})  u $.
\end{proof}
Note that Lem.~\ref{le:gauge_invariance} holds also in case of a smooth non-circular boundary.

Our results are stated in the Aharonov--Casher gauge. It is important to stress out, that this is just for the convenience in the computations. Of course, the spectral results are independent of this gauge choice as established by Lem.~\ref{le:gauge_invariance}.

\section{Main results} \label{chap:main_thms}
Using the setup introduced in Sec.~\ref{sec:problem_setting}
we are in a position to state the main theorems of this paper.
The first result concerns an Aharonov--Casher type
theorem for Dirac operators on $\Rb^2$ with circular holes. 
More precisely we will prove
\begin{theorem} \label{thm:unbdd}
	Let $M  = \Cb \setminus \cup_{k\leq N} \Omega_k$ and let
	 $D_a^{\APS}$ be the Dirac operator 
	with the magnetic field \eqref{eq:tot_mag_field}
	in the gauge~\eqref{eq:h_relation}, \eqref{eq:potential}.
	If $\Phi\neq 0$ then there are
	\begin{align*}
		\floor{\frac{|\Phi|}{2\pi}} 
	\end{align*}
	zero modes of the operator  $D_a^{\APS}$ 
	with the APS boundary conditions~\eqref{eq:BC_hole_j} 
	on the inner components of the boundary. 
	These states have spin up (\ie are eigenvectors of $\sigma_3=\mathrm{diag }(1,-1)$ with eigenvalue $+1$) if $\Phi>0$ and spin down (\ie are eigenvectors of $\sigma_3$ with eigenvalue $-1$) if $\Phi<0$.
	If $\Phi = 0$ then there are no zero modes.	
\end{theorem}
\begin{corollary} \label{cor:zero_in_ess_spec}
	Under the assumptions of Thm.~\ref{thm:unbdd}, zero is an eigenvalue embedded in the essential spectrum of 
	the operator $D_a^{\APS}$.
\end{corollary}
\begin{proof}
	Due to the unboundedness of $M$ and the existence of the zero modes 
	we can construct a Weyl sequence as follows. Let
	$\mathbb{B}$ denote a ball centred at the origin such that 
	$\cup_{k\leq N} \Omega_k \subset \mathbb{B}$. 
	Consider a compactly supported radial smooth function $\phi$ on 
	$\Cb \setminus \mathbb{B}$.
	If $\Phi \neq 0$, choose a zero mode $u\in \dom(D_a^{\APS})$.
	Denote $\phi_n(z)\coloneqq \frac{1}{n}\phi (\frac{z}{n})$.
	Then $u_n(z)\coloneqq u(z)\phi_n(z) \subset \dom(D_a^{\APS})$ tends weekly to zero, while 
	\begin{align*}
		D_au_n = (D_a u)\phi_n + u D_0 \phi_n 
			= u D_0 \phi_n \,.
	\end{align*}
	Using that in polar coordinates $D_0$ acts as 
	$D_0 
		= -\ii 
		\begin{pmatrix}
	 		0 					& \ee^{-\ii \theta} \\
	 		\ee^{\ii \theta} 	&  0 
	\end{pmatrix} 
	\p_r $ 
	on radial functions, 
	one concludes $\|D_a u_n\| \rightarrow 0$ as $n\rightarrow \infty$.
	
	If $\Phi = 0$, 
	there are no zero modes.
	In view of Rem.~\ref{rem:a_is_bounded} the vector potential corresponding to the smooth compactly supported magnetic field in the bulk is bounded. Moreover, 
	thanks to Le.~\ref{le:gauge_invariance} we can, without loss of generality, consider that the magnetic field inside the holes is formed by the Aharonov Bohm fluxes. 
	The vector potential of such field is bounded outside of $\mathbb{B}$.
	Consequently the proof works with sequence 
	$u_n = \phi_n$.
\end{proof}

For the case when the underlying manifold is a disc with holes, the number of zero modes differs. In particular for the total flux in the range
$(\pi, 2\pi] \,\mathrm{ mod }\, 2\pi k$, for $0\leq k \in\Zb$ or
$[-2\pi, -\pi]\, \mathrm{ mod }\, 2\pi k'$, for $0\geq k' \in\Zb$,
 we obtain 
an extra zero mode as opposed to the unbounded case, see Remark~\ref{rem:zero_modes_result_form}.
\begin{theorem} \label{thm:bdd}
	Let $M  = \Omega _{out} \setminus \cup_{k\leq N} \Omega_k$ 
	and let $D_a^{\APS}$ be the Dirac operator with the magnetic field \eqref{eq:tot_mag_field}
	in the gauge~\eqref{eq:h_relation}, \eqref{eq:potential}.
	 Then there are
	\begin{align*}
		\left| \floor{\frac{\Phi}{2\pi}+ \frac{1}{2}}\right|
	\end{align*}
	zero modes of the operator  $D^{\APS}_a$ with
	 the APS boundary conditions~\eqref{eq:BC_hole_j}
	on the inner components and~\eqref{eq:BC_hole_N+1} on the outer component of the boundary.
	 In particular, there are no zero modes in the case $\Phi \in (-\pi, \pi]$.
	If $\Phi > 0$ then all the zero modes have spin up. 
	If $\Phi < 0$ then they have spin down.
\end{theorem}

\begin{remark}\label{rem:zero_modes_result_form}
	\begin{itemize}
	\item
		The particular form of the (non-normalised) 
		zero modes of  the Dirac operator in the 
		Aharonov--Casher gauge with the APS boundary condition
		is also known from the proof. 
		Depending on the sign of the total flux $\Phi$, 
		they are purely spin up or purely spin down
	 	\begin{align*}
	 		\begin{pmatrix}
	 			u^+ 	\\
	 			0
			\end{pmatrix} \,, \quad
			\begin{pmatrix}
	 			0	\\
	 			u^-
			\end{pmatrix} \,,  		
		\end{align*}	
		where 
		\begin{align*}
	 		u^+(z) = \ee ^{h (z)} 
	 			\sum _{0\leq n <\frac{\Phi}{2\pi} - 1} a_n z^n	\,, \quad
	 		u^-(z)  = \ee ^{-h (z)} 
	 			\sum _{0\leq n <-\frac{\Phi}{2\pi} -1} b_n \overline{z}^n \,,
	 		\quad \text{ if } \M = \Cb \,,
		\end{align*}
		and,
		\begin{align}\label{eq:zm_bdd}
	 		u^+(z) = \ee ^{h (z)} 
	 			\sum _{0\leq n < \frac{\Phi}{2\pi} 
	 				- \frac{1}{2}} a'_n z^n	\,, \quad
	 		u^-(z)  = \ee ^{-h (z)} 
	 			\sum _{0\leq n \leq -\frac{\Phi}{2\pi} 
	 				-\frac{1}{2}} b'_n \overline{z}^n \,,
	 		\quad \text{ if } \M =\Omega _{out} \,,
		\end{align}
		with some coefficients $a_n, b_n, a'_n, b'_n \in \Cb$.
	\item
	 	Notice that for certain values of fluxes there 
	 	is an extra zero mode for the bounded region.
		This happens in particular if there exists 
		an integer $k$ such that
	 	\begin{align}\label{eq:extra_zm_fluxes}
	 	 	\frac{\Phi}{2\pi} 
	 	 		= k + \epsilon 
	 	 			\quad
	 	 			\text{ with }
	 	 		\epsilon \in 
	 	 		\begin{cases}	
	 	 			(1/2, 1] 		& \text{ if } k\geq 0 	\\
	 	 			[-1, -1/2] 	& \text{ if } k\leq 0 
	 	 		\end{cases} \,.
	 	\end{align}
	 	If we let the radius of the disc $\Omega_{out}$ 
	 	grow to infinity we can compare the asymptotics of the zero
	 	modes~\eqref{eq:zm_bdd} as $|z|\rightarrow \infty$ 
	 	and easily check that the zero mode with the highest power 
	 	of $z$ in $g^+$ 
	 	(or $\overline{z}$ in $g^-$ in case of spin down zero modes) 
	 	satisfying the APS boundary condition would not 
	 	be square integrable at infinity exactly for the values 
	 	of fluxes in \eqref{eq:extra_zm_fluxes}. 
	 	They exhibit the behaviour $u^{\pm}(z)\sim |z|^{\mp \epsilon_{\pm}}$ as $|z|\rightarrow \infty$ with $\epsilon_+ \in (1/2, 1]$ and $\epsilon_- \in [-1, -1/2]$. Therefore they are decaying at infinity and thus are resonances of the Dirac operator in 
	 	the unbounded case when 
	 	$\M = \Cb$, \ie the operator from Thm.~\ref{thm:unbdd}.
	 	To complete the picture we recall that 
	 	for values $\Phi/2\pi \in (-1/2, 1/2]$ there are 
	 	no zero modes on either plane or a disc with holes.
	\end{itemize}
\end{remark}

The standard argument of  Aharonov and Casher from \cite{AC79} will be very useful for showing Thms.~\ref{thm:bdd} and \ref{thm:unbdd}. We will first use it to prove the following
\begin{prop} 	\label{prop:zero_modes_form}
	Let $D^{\max}_a$ be the maximal extension~\eqref{eq:dom_max} of the
	 Dirac operator on $M$ 
	 with the magnetic field \eqref{eq:tot_mag_field}
	in the gauge~\eqref{eq:h_relation}, \eqref{eq:potential}
	and let $D^{\max}_a u = 0$ for some $u\in \dom(D_a^{\max})$. Then 
	$u$ is of the form $u =(u^+, u^-)$ with
	\begin{align*}
 		u ^{\pm} = \ee ^{\pm h} g ^{\pm} \,, 
	\end{align*}
	where $g^+$ is analytic on $M$ and $g^- $ is anti-analytic on $M$.
	Vice versa, any function of this form such that $u\in \dom(D_a^{\max})$ satisfies $D^{\max}_a u=0$.
\end{prop}
\begin{proof}
	To find solutions of $D^{\max}_a u = 0$, $u=(u^+, u^-) \in L^2(M; \Cb^2)$ 
	we rewrite the problem using the Aharonov--Casher gauge~\eqref{eq:h_relation}
	as
	\begin{align*}
 		 \ee ^{h} \p _{\overline{z}} (\ee ^{-h} u^+) = \left[\p_{\overline{z}} - \frac{\ii a}{2}\right] u^+	= 0  \,,
 		\quad
 		\ee ^{-h} \p _{z} (\ee ^{h} u^-) =\left[\p_z - \frac{\ii \overline{a}}{2} \right]u^{-} 	= 0 \,.
	\end{align*}
	This is satisfied if and only if the function $g^+\coloneqq \ee^{- h} u^{+}$ 
	is analytic and $g^-\coloneqq \ee^{ h} u^{-}$ is anti-analytic on $M$.
\end{proof}
\begin{remark}
	To avoid any confusions, we would like emphasise, that by taking $u\in L^2(M;\Cb^2)$ in the previous proof, we are not claiming that this is the maximal domain $\dom(D^{\max}_a)$. If $u\in L^2(M; \Cb^2)$ such that $D^{\max}_au=0$ it is by definition in the maximal domain.
\end{remark}
\subsection{Proof for unbounded region with holes}	\label{sec:zero_modes_holes}
The main idea of the proofs of Theorems~\ref{thm:bdd} and \ref{thm:unbdd} is to show that
if $u \in \dom (D^{\APS}_a)$ 
and hence, $(u^+, u^-)^T$ satisfies the APS boundary condition \eqref{eq:BC_hole_j},
then the functions $g^+$ and $g^-$ from Prop.~\ref{prop:zero_modes_form} can be extended analytically in $z$ and $\bar{z}$ respectively 
inside the holes of $M$.
In the following example we demonstrate, that for the case of a single hole with a magnetic field inside the hole this extension is a straightforward process.

\begin{small}
\begin{example}
	For one hole in the plane we worked out the APS boundary condition in Ex.~\ref{ex:one_hole}. Here we use the set up and notation from that example. For the functions $g^{\pm}$ from 
	Prop.~\ref{prop:zero_modes_form} we can write the Laurent series
	\begin{align*}
	 	g^{+} = \sum_{n \in \Zb} a_n z^n 
	 	\qquad
	 	\text{ and }
	 	\qquad 
	 	g^{-} = \sum_{n \in \Zb} b_n \bar{z}^n \,, 	
	\end{align*}
	with some complex coefficients $a_n, b_n$.
	Taking into account~\eqref{eq:h_spherical_field}, the formal restriction of $u^{\pm} = \ee^{\pm h} g^{\pm}$ to the boundary $\p \Omega$ then reads
	\begin{align*}
	 	\gamma_0 u^{+} = R^{-\Phi/2\pi}\sum_{n \in \Zb} a_n R^n \ee^{\ii n \phi} 
	 	\qquad
	 	\text{ and }
	 	\qquad
	 	\gamma_0 u^- = R^{\Phi/2\pi} \sum_{n \in \Zb} b_n R^n  \ee^{-\ii n \phi} \,. 	
	\end{align*}
	Here $\gamma_0$ is the trace map as in~\eqref{eq:trace_map}.
	By Le.~\ref{le:gauge_invariance} we can restrict ourselves to 
	$\Phi \in [-\pi, \pi)$. 
	The boundary condition~\eqref{eq:BC_one_hole} then yields 
	that $a_n = 0 =b_n$ for $n< 0$.
	Consequently, $g^+$ has to be analytic and $g^-$ anti-analytic on $\Cb$.

	We leave the argument that the trace $\gamma_0$ is indeed given by the formal restriction to the boundary to the general case, which is worked out in Sec.~\ref{ssec:trace_of_efuns}
\end{example}
\end{small}

While the argument in our example is rather simple, extending $g^{\pm}$
(anti-)analytically inside the holes requires a new approach if we have several holes.
Let us fix an index $j$ 
and denote $\mathcal{A}$ an open annulus co-centred with the hole $\Omega_j$
of inner radius $R_j$ and 
outer radius $R$ such that $\mathcal{A} \cap \supp B = \emptyset$
with $B$ as in~\eqref{eq:tot_mag_field}.
In particular the scalar potential $h(z)$ is bounded on $\mathcal{A}$.

Recall that $g^+$ is analytical 
and $g^-$ is anti-analytical on $M$.
 Thus on $\mathcal{A}$ 
 they have the Laurent series
\begin{align} \label{eq:g_series}
	g^+(z)	= \sum_{n \in \Zb} a_n (z-w_j)^n \,,
	\qquad
	g^-(z)	= \sum_{n \in \Zb} b_n \overline{(z-w_j)}^n \,,
\end{align}
with some $a_n, b_n \in \Cb$.
To find the boundary values of $u ^{\pm}$
and compare them to the boundary condition \eqref{eq:BC_hole_j},
it is convenient to introduce the following functions  
for $z \in \mathcal{A}\cup \overline{\Omega_j}\setminus \{w_j\} $
\begin{align} \label{eq:G_j}
	\begin{split}
	 	G^+_j(z) 
	 		&\coloneqq - \ii \int_{\gamma(z_{0j}, z)} \vec{a} \dd \vec{s} +\int_{\gamma(z_{0j}, z)} \frac{\Phi'_j}{2\pi (z'-w_j)} \dd z' \,,	\\
	 	G^-_j(z) 
	 		&\coloneqq - \ii \int_{\gamma(z_{0j}, z)} \vec{a} \dd \vec{s} -\int_{\gamma(z_{0j}, z)} \frac{\Phi'_j}{2\pi \overline{(z'-w_j)}} \dd \bar{z}' \,,
	\end{split}
\end{align}
where by $\gamma(z_{0j}, z)\subset \mathcal{A}\cup \overline{\Omega_j} \setminus \{w_j\}$ 
we denoted the path of integration with the endpoints 
$z _{0j} = w_j+R_j$ and $z \in \mathcal{A}\cup \overline{\Omega_j}\setminus \{w_j\}$.
We stress that throughout the whole section we are, 
owing to Lem.~\ref{le:gauge_invariance}, 
assuming $B\big|_{\Omega_j} = B_j = \Phi'_j \delta_{w_j}$, 
with the normalised flux $\Phi'_j \in [-\pi, \pi)$.
This, further, allows us to extend the definition of the vector potential $a$
that is given by~\eqref{eq:h_relation} inside the region $\Omega_j \setminus \{w_j\}$.

The definition \eqref{eq:G_j} is motivated by
the fact that the restrictions of $G ^{\pm}_j(z)$ to the 
		boundary $\p \Omega_j$ satisfy
		\begin{align*} 
 			G^{\pm}_j(z) \mid_{z \in \p \Omega_j} 
 				= - \ii \int_{\gamma_j} \vec{a} \dd \vec{s} 
 					+ \ii \frac{\Phi'_j}{2\pi} \phi_j \,,
		\end{align*}
		where $\gamma_j \subset \p \Omega_j$ is the curve connecting 
		the points $z_{0j} = w_j + R_j$ and $z$ counter-clockwise.
The lemma below further shows that $G^{\pm}_j(z)$ are well defined 
on $\mathcal{A}\cup \overline{\Omega_j}\setminus \{w_j\}$.	
\begin{lemma} \label{lemma:path_indep_2}
	$G^{\pm}_j(z)$ are independent of the choice of the path $\gamma (z_{0j}, z)$ contained in 
	$\mathcal{A}\cup \overline{\Omega_j} \setminus \{w_j\}$.
\end{lemma}
\begin{proof}
	We show the equivalent statement that $G^{\pm}_j(z) = 0$ for any loop 
	$\gamma = \gamma(z_{0j}, z= z_{0j}) \subset 
	\mathcal{A}\cup \overline{\Omega_j} \setminus \{w_j\}$. 
	By definition of the flux  the first summands on the right-hand sides of \eqref{eq:G_j} read
	\begin{align*}
 		-\ii\int_{\gamma} \vec{a} \dd \vec{s} 
 			=  -\ii\ell \Phi'_j \,,
	\end{align*}
	where $\ell \in \Zb$ is the winding number of the loop $\gamma$ around the point $w_j$.
	The result then follows from the definition of the winding number
	\begin{align*}
		\ell 
			\coloneqq \frac{1}{2\pi \ii}\int_{\gamma} \frac{1}{z'-w_j} \dd z'
			= \frac{-1}{2\pi \ii}\int_{\gamma} \frac{1}{\overline{z'-w_j}} \dd \bar{z}' \,.
	\end{align*}
\end{proof}
Now we show that $G_j^\pm$ satisfy another important property.
\begin{prop}\label{cor:exp(F)}
Define
		\begin{align} \label{eq:F_pm}
			F^{\pm}_j(z) \coloneqq \pm h(z) + G^{\pm}_j(z) \,,
		\end{align}
		with $h$ given by \eqref{eq:potential}.
	The functions $F^+_j(z)$ and $F^-_j(z)$ are 
	analytic on $\mathcal{A}\cup \overline{\Omega_j}$ in $z$ and $\bar{z}$ respectively.
	And, in particular, there are the following series of the exponentials 
	\begin{align*}
 		\ee ^{F_j ^{+}}  =  \sum _{k\geq 0} c^+_k (z-w_j)^k	\,, \quad
 		\ee ^{F_j ^{-}}  =  \sum _{k\geq 0} c^-_k \overline{(z-w_j)}^k \,,
	\end{align*}
	with $c_0 ^{\pm} \neq 0$.
\end{prop}

Before we prove this statement let us present a lemma.
	\begin{lemma} \label{rem:analyticity_of_path_indep_integral}
		Let $g$ be defined on a domain in $\Cb$. If either
		\begin{align*}
	 		g_{an}(z)
	 			&\coloneqq \int_{\gamma} g(w) \dd w 
	 			= \int_{\gamma} (g_1 + \ii g_2) (\dd t_1 + \ii \dd t_2) \,,		\\
	 			\text{or} 	\\
	 		g_{anti}(z)
	 			&\coloneqq \int_{\gamma} g(w) \dd \bar{w} 
	 			= \int_{\gamma} (g_1 + \ii g_2) (\dd t_1 - \ii \dd t_2) \,,
		\end{align*}
		 are independent of the path $\gamma$
		connecting a fixed point $z_0 \in \Cb$ and a point $z \in \Cb$,
		then $g_{an}$ or $g_{anti}(z)$
		are analytic in $z$ or $\bar{z}$ respectively.
\end{lemma}
\begin{proof}		
		Due to the independence on path $\gamma$ it holds 
		$\p_x \int_{\gamma}(g_1+\ii g_2) \dd t_2 = 0$ and 
		$\p_y \int_{\gamma}(g_1+\ii g_2) \dd t_1 = 0$.
		So we have
		\begin{align*}
	 		\p_x  \int_{\gamma} g(w) \dd w 
	 			&= \p_x \int_{\gamma} (g_1 + \ii g_2) \dd t_1
	 			= g(z) 	\\
	 		\p_y  \int_{\gamma} g(w) \dd w 
	 			&= \p_y \int_{\gamma}\ii  (g_1 + \ii g_2) \dd t_2
	 			= \ii g(z) 	\,,
	 	\end{align*}
	 	where in the second equalities we used the fundamental theorem 
	 	of calculus.
	 	Hence, \\
	 	$\frac{1}{2} (\p_x \pm \ii \p_y) \int_{\gamma} g(w) (\dd t_1 \pm \ii \dd t_2) =0$.
\end{proof}
 \begin{proof}[Proof of Prop.~\ref{cor:exp(F)}]
 	The analyticity follows from the fact, which will be proved below,
 	stating that $F_j ^{\pm}$ have the following forms on $\mathcal{A}\cup \overline{\Omega_j} \setminus \{w_j\}$ 
	\begin{align}\label{eq:F_pm_2}
		\begin{split}
	 		F^+_j(z)  
	 			& =  h(z_{0j}) + \int_{\gamma(z_{0j}, z)} \sum_{\substack{k \leq N \\k\neq j}} 
	 			\left(2\p_{z'} h_{[k]}\right)  \dd z' \,, 	\\
	 		F^-_j(z)  
	 			& =  -h(z_{0j}) - \int_{\gamma(z_{0j}, z)} \sum_{\substack{k\leq N \\k\neq j}} 
	 			\left( 2\p_{\bar{z}'} h_{[k]}\right)  \dd \bar{z}' \,,
		\end{split}
	\end{align}
	where $h_{[k]} = \frac{-\Phi'_k}{2\pi} \log |z-w_k|$ for $z \neq w_k$ is the 
	scalar potential of the field $B_k$ inside the hole $\Omega_k$.
	A direct  computation
	(or the ``defining'' Poisson equation $-\Delta h_{[k]} = B_k$ for $h_{[k]}$) 
	yields that  the integrands 
	$ \sum_{k\neq j} 2\p_{z} h_{[k]}$ and 
	$ \sum_{k\neq j} 2\p_{\bar{z}} h_{[k]}$
	 are analytic on $\mathcal{A}\cup \overline{\Omega_j}$ in $z$ and $\bar{z}$ respectively.
	 It follows from Lem.~\ref{rem:analyticity_of_path_indep_integral}, that this indeed implies analyticity of $F_j^\pm$.

	Now we will show that the equalities \eqref{eq:F_pm_2} hold.
	To that end we use the relation \eqref{eq:h_relation}, \ie
	\begin{align*}
		a_x = \p_y h \,,	 \quad 	 a_y = -\p_x h 	\,,	
	\end{align*}
	and write with the aid of the fundamental theorem of calculus
	\begin{align*}
 		h(z) = h(z_{0j}) + \int_{\gamma(z_{0j}, z)} \p_x h \dd x+ \p_y h \dd y \,,
	\end{align*} 
	where 
	$\gamma(z_{0j}, z)\subset 
		\mathcal{A}\cup \overline{\Omega_j} \setminus \{w_j\}$
	is an arbitrary path connecting $z_{0j}$ and $z$.
	Thus we get
	\begin{align*}
 		h - \ii \int_{\gamma(z_{0j}, z)}  \vec{a} \dd \vec{s}
 			& = h(z_{0j}) + \int_{\gamma(z_{0j}, z)} \p_x h \dd x+ \p_y h \dd y -\ii \p_y h \dd x + \ii \p_x h \dd y 	\\
 			& = h(z_{0j}) + 2 \int_{\gamma(z_{0j}, z)} \p_{z'} h \dd z' \,,
	\end{align*}  
	and similarly (or by complex conjugation), 
	$-h - \ii \int_{\gamma(z_{0j}, z)}  \vec{a} \dd \vec{s}
 			 = -h(z_{0j}) - 2 \int_{\gamma(z_{0j}, z)} \p_{\bar{z}'} h \dd \bar{z}'$.
 			 
	Finally, we recall that the concrete form \eqref{eq:potential} 
	 of the potential function for 
	 $B_j = \Phi'_j \delta _{w_j}$ is
	 $h_{[j]} = \frac{-\Phi'_j}{2\pi} \log|z-w_j|$
	and compute	the derivatives
	\begin{align*}
 		\p_z h_{[j]}(z)
 			=   \frac{-\Phi'_j}{4\pi} \p_z \log|z-w_j|^2
 			 = -\frac{1}{4\pi} \frac{\Phi'_j}{z-w_j}  \,, 	
 			 \quad \text{ and } \quad
 		\p_{\bar{z}} h_{[j]}(z)
 			= -\frac{1}{4\pi} \frac{\Phi'_j}{\overline{z-w_j}}  \,,
	\end{align*}
	which together with the definitions \eqref{eq:F_pm} 
	and \eqref{eq:G_j} give \eqref{eq:F_pm_2}.	
 \end{proof}
 
\subsubsection{Trace of the eigenfunctions}
\label{ssec:trace_of_efuns}
Using the properties of $G_j^{\pm}$
we are able to find the trace $\gamma_0$ of 
functions that have a form of the zero modes
on $\mathcal{A}$.
For brevity we will denote by $\check{H}_j$ the space $\check{H}(A\big|_{\p\Omega_j})$,
recall~\eqref{eq:H_check} for the definition of the norm
on $\check{H}(A)$. 

\begin{lemma}\label{le:conv_in_max_dom}
	Consider the function $u=(u^+, u^-)^T$
	which has a convergent Laurent series on $\mathcal{A}$
	\begin{align*}
		u
		= \sum_{n\in \Zb} (\ee^{h}  a_n (z-w_j)^n, 
			\ee^{-h} b_n (\overline{z-w_j})^n)^T 
			\in L^2(\mathcal{A}, \Cb^2) \,, 
			&& a_n, b_n \in \Cb\,.
	\end{align*}	
	Then we have the convergence
	\begin{align}
		u^Q
		=
	 	\begin{pmatrix}
	 	 	u^{+, Q}\\
	 	 	u^{-, Q}
	 	\end{pmatrix}
	 	\coloneqq
	 	\sum_{|n|\leq Q}
	 	\begin{pmatrix}
	 	 	\ee^{h}  a_n (z-w_j)^n	\\
	 	 	\ee^{-h} b_n (\overline{z-w_j})^n
	 	\end{pmatrix}
	 	\rightarrow 
	 	\begin{pmatrix}
	 	 	u^+ 	\\
	 	 	u^-
	 	\end{pmatrix} \,, 
	 	\quad
	 	\text{ as } 
	 	Q\rightarrow \infty
	\end{align}
	in the operator graph norm 
	$\|\cdot\|_{D_a, \mathcal{A}} \coloneqq \|\cdot\|^2_{L^2(\mathcal{A}, \Cb^2)}+ \|D_a(\cdot)\|^2_{L^2(\mathcal{A}, \Cb^2)}$.
\end{lemma}
\begin{proof}
	By Prop.~\ref{prop:zero_modes_form} it holds that
	$D_a u^{Q} = D_a u = 0$ on $\mathcal{A}$.
	Hence the operator graph norm of $u$ has contribution only 
	from $\|u\|^2_{L^2(\mathcal{A}, \Cb^2)} = \|u^+\|^2_{L^2(\mathcal{A})} + \|u^-\|^2_{L^2(\mathcal{A})}$. We compute the first summand.
	\begin{align*}
	 	\|u^{+}\|^2_{L^2(\mathcal{A})}
	 		= 2\pi \int_{R_j}^{R} r\dd r 
	 			\left( \ee^{2h} \sum_{n\in \Zb} |a_n|^2 r^{2n}
	 				\right) 	
	 		\geq C
	 		\left(|a_{-1}|^2\ln\frac{R}{R_j}
	 		+\sum_{n\neq -1} |a_n|^2 \frac{R^{2n+2} - R_j^{2n+2}}{2n+2} \right) \,,
	\end{align*}
	where the constant $C =2\pi\min_{z\in \mathcal{A}} \ee^{2h(z)}$ is non-zero.
	This shows that the sum on the right-hand side
	is convergent and therefore
	\begin{align*}
	 	\|u^{+} - u^{+, Q}\|^2_{L^2(\mathcal{A})}
	 		\leq 2\pi \max_{z\in \mathcal{A}} \ee^{2h(z)}
	 			\sum_{n>Q} |a_n|^2 \frac{R^{2n+2} - R_j^{2n+2}}{2n+2} 		
	 				\rightarrow 0 \,,
	 	\quad
	 	\text{ as } 
	 	Q\rightarrow \infty \,.
	\end{align*}
	The proof for $u^-$ is a matter of substituting $\ee^{2h}$ by 
	$\ee^{-2h}$ and the coefficients $a_n$ by $b_n$ in the computation above.
\end{proof}
As a corollary we obtain from \cite[Thm.~1.7]{BB12}
the convergence in $\check{H}_j$ of the traces of $u^Q$ to the trace of the zero mode $u$. 

\begin{lemma}\label{le:traces}
	Let $u^{Q}$ be as in Lem.~\ref{le:conv_in_max_dom}.
	In $\check{H}_j$ we have the following convergence of the traces $\gamma_0(u^{Q})$
		\begin{align} \label{eq:traces_j}
		\lim_{Q\rightarrow \infty}
		 	\gamma_0 
		 	\begin{pmatrix}
		 	 	u^{+, Q}	\\
		 	 	u^{-, Q}
		 	\end{pmatrix}
		 =
		 	\sum_{n\in \Zb} \sum_{k\geq 0} R_j^{n+k}
		 	\begin{pmatrix}
			 	a_n c_k^+ \psi^j_{n+k}	
			 	\\
			 	b_n c_k^- \psi^j_{-(n+k)}			
		 	\end{pmatrix} 
		 \eqqcolon
		 	\begin{pmatrix}
		 	 	u_0^{+}	\\
		 	 	u_0^{-}
		 	\end{pmatrix}	\,.
		\end{align}
		The vectors $\psi_{\ell}^j$ 
		were introduced in~\eqref{eq:eigenproblem_bdry_operator}
		and $c_{k}^{\pm}$ are the coefficients from 
		Prop.~\ref{cor:exp(F)}.
\end{lemma}
\begin{proof}
	First we show that the sum in~\eqref{eq:traces_j} is an element of $\check{H}_j$.
	Note that due to the definite chirality 
 	of eigenvectors of $A|_{\p \Omega_j}$ we have $\|(v^+, v^-)\|^2_{\check{H}_j} = \|(v^+,0)\|^2_{\check{H}_j} + \|(0, v^-)\|^2_{\check{H}_j}$ for any 
 	$(v^+, v^-)^T\in \check{H}_j$. 
 	Therefore we work out the proof only for the spin up part and leave out the spin down part which is analogous.
 	Let $u^+$ be as in Lem.~\ref{le:conv_in_max_dom}.
	Then boundedness of $\ee^{G_j^+}$ on $\mathcal{A}$ implies
	$\ee^{G_j^+} u^+ \in L^2(\mathcal{A})$. 
	We compute its norm explicitly
 	\begin{align*}
 		\|\ee^{G_j^+} u^+ \|_{L^2(\mathcal{A})}^2
 			&=
	 	 	\int_{R_j}^{R} \dd r \int_0^{2\pi} r\dd \phi 
 	 		\bigg
 	 			|\ee^{G_j^+ + h}g^+
 	 		\bigg|^2 
 	 		=
 	 		\int_{R_j}^{R} \dd r \int_0^{2\pi} r\dd \phi
 	 		\bigg|
 	 			\sum_{\stackrel{k\geq 0}{ n\in \Zb}} 
 	 				a_n c_k^+ r^{k+n}\ee^{\ii \phi (k+n)}
 	 		\bigg|^2 
 	 		\\
 	 		&=
 	 		2\pi 
 	 		\bigg|
 	 			\sum_{\stackrel{k+n=-1}{ k\geq 0}}
 	 			a_n c_k^+ 
 	 		\bigg|^2 
 	 		\ln\frac{R}{R_j}
 	 		+
 	 		2\pi \sum_{\ell\neq -1}
 	 			\bigg|\sum_{\stackrel{k+n=\ell}{ k\geq 0}}
 	 			a_n c_k^+ \bigg|^2
 	 			\frac{1}{2 \ell+2} (R^{2 \ell+2} -R_j^{2\ell+2}) 
 	 		\,,
 	\end{align*}
 	where in the second equality we used that 
 	$\ee^{G_j^+(z)+ h(z)}=\ee^{F_j^+}$ is the analytic function on 
 	$\mathcal{A}\cup \overline{\Omega_j}$
 	from 
 	Prop.~\ref{cor:exp(F)}.
 	This expression can be used to bound the $\check{H}_j$ norm of the spin up component
 	of~\eqref{eq:traces_j}
 	\begin{align}\label{eq:limiting_function_trace}
 	 	\|(\sum_{\ell\in \Zb} 
 	 	\sum_{\stackrel{n,k}{n+k=\ell}} R_j^{\ell}
			 	a_n c_k^+ \psi^j_{\ell} \,, 0)^T\|^2_{\check{H}_j}
	 		= 
 	 		\sum_{\ell\in \Zb} 
 	 		\bigg|
 	 			\sum_{\stackrel{n+k=\ell}{k\geq 0}} R_j^{\ell}
			 	a_n c_k^+ \psi^j_{\ell}
			 \bigg| (1+ |\lambda^j_{\ell}|^2)^{s} \,,
 	 		\quad 
 	 		\begin{cases} 
 	 			s= 1/2 & \text{ if } \lambda^j_{\ell}< 0 \\
 	 			s= -1/2 & \text{ if } \lambda^j_{\ell}\geq 0 
 	 		\end{cases} \,,
 	\end{align}
 	with $\lambda^j_{\ell}$ from~\eqref{eq:eigenproblem_bdry_operator}.
	Indeed, observing that
 	\begin{enumerate}
 		 	\item If $K\in [-1, 0)$, 
 		 		then for any $n\geq 0$ and $C<1/\sqrt{5}$
					it holds
			 		\begin{align*}
			 	 		n+1\geq C\sqrt{(n-K)^2+1} \,, \qquad
			 	 		\frac{1}{n+1}\geq C \frac{1}{\sqrt{(n+K)^2+1}} \,.
			 		\end{align*}
 	\item There exist constants $C_{<,>,0}$ such that 
 		$\ln \frac{R}{R_j} \geq \frac{C_{0}}{R^2_j}$, and
 	\begin{align*}
 	 	 R^{n} -R_j^{n}
 	 	 	 		&\geq  \left(R R_j^{-1}\right)^{n}
 	 	 	 			 R_j^{n}
 	 	 	 			\left(1- \left(R_j R^{-1}\right)^{n} 
 	 	 	 				\right)	
 	 	 	 		\geq C_{>} n^{2} R_j^{n} \,,
 	 	 	 		& \text{if } n>0 \,, \\
 	 	 R_j^{n} -R^{n}
 	 	 	 		&=R_j^{n} 
 	 	 	 			\left(1- \left(R_j R^{-1}\right)^{|n|}
 	 	 	 				\right)	
 	 	 	 		\geq   C_{<}  R_j^{n}\,,
 	 	 	 		& \text{if } n<0 \,, \\
 	 	 	\end{align*}
 	 	 	\end{enumerate}
 	we deduce 
 	$\|\ee^{G_j^+} u^+ \|_{L^2(\mathcal{A})}^2
 	 		\geq const \|(u_0^+, 0)^T\|^2_{\check{H}_j}$.
 	
 	 		Applying Prop.~\ref{cor:exp(F)} and using the definition
 	 		of $\psi_{\ell}^j$
 	 		we obtain
 	 		\begin{align*}
 	 		 	\gamma_0[(u^{+, Q}, 0)^T]
 	 		 		= \sum_{|n|\leq Q}
 	 		 			R_j^n a_n 
 	 		 			\ee^{\ii n \phi_j-G_j^+(z)} \ee^{(h+G_j^+)(z)}\big|_{z\in \p \Omega_j}
 	 		 		= \sum_{|n|\leq Q}
 	 		 			R_j^{n+k} a_n c_k^+ \psi_{n+k}^j
 	 		\end{align*}
 	 which with the convergence of the sum in~\eqref{eq:traces_j}
 	 in $\check{H}_j$ norm finishes the proof for the spin up component.	
 	 Similarly one shows
 	$\|\ee^{G_j^-} u^- \|_{L^2(\mathcal{A})}^2
 	 		\geq const \|(0, u_0^-)^T\|^2_{\check{H}_j} $
 	 and the convergence in the spin down component of~\eqref{eq:traces_j}.
\end{proof}

With these technical preliminaries we are ready to present
a key statement for the proof of Thm.~\ref{thm:unbdd}.
\begin{prop}\label{cor:g_is_analytic}
	Let $(u^+, u^-)$ be a zero mode of the Dirac operator $D_a^{\APS}$
	with the magnetic field~\eqref{eq:tot_mag_field},
	that satisfies the APS boundary condition~\eqref{eq:BC_hole_j} on $\p \Omega_j$.
	Then the functions $g^+$ and $g^-$ from Prop.~\ref{prop:zero_modes_form}
	 can be analytically extended inside the region 
	 $\Omega_j$ in $z$ and $\overline{z}$ respectively.
\end{prop}
\begin{proof}
	On $\p \Omega_j$ we compare the trace given by Lem.~\ref{le:traces}
	with the APS boundary condition~\eqref{eq:BC_hole_j}.
	We remind that we are using the normalized fluxes through the holes
	(see Lem.~\ref{le:gauge_invariance}) which satisfy 
	$\Phi'_j \in [-\pi, \pi)$
	and therefore
	\begin{align*}
	 	\gamma_0(u^+, u^-)
	 	= 
	 		\sum_{n\in \Zb} \sum_{k\geq 0} R_j^{n+k}
		 	\begin{pmatrix}
			 	a_n c_k^+ \psi^j_{n+k}	
			 	\\
			 	b_n c_k^- \psi^j_{-(n+k)}			
		 	\end{pmatrix} 
	 	= 
	 		\begin{pmatrix}
	 			\sum_{\ell\geq 0} \beta^+_{\ell} \psi_{\ell}^j	\\
	 			\sum_{\ell\leq 0} \beta^-_{\ell} \psi_{\ell}^j
	 		\end{pmatrix} \,.
	\end{align*}
	Here $\beta^{\pm}_{\ell}\in \Cb$ are some constant coefficients.
	Hence
	$a_n c_k^+ = b_n c_k^- = 0$ whenever $n+k<0$.
	In particular, since $c_0^{\pm}\neq 0$, it holds
	$a_n = b_n = 0$ for all $n<0$.
	This in turn means that $g^{+}$ and $g^-$ can be analytically extended inside $\Omega_j$ in variables $z$ and $\overline{z}$ respectively. 	 
\end{proof}

Applying now the $L^2$ integrability condition at infinity to $u ^{\pm}$
the Aharonov--Casher type result for our setting $\mathscr{M} = \Cb$ follows.

\begin{proof}[Proof of Thm.~\ref{thm:unbdd}] 
By Prop.~\ref{cor:g_is_analytic} the zero modes are of the form
\begin{align*}
	u =
	\left (
		\ee ^{h} \sum _{n = 0} ^{n^+} a _{n} z ^{n}\,, 
 		\ee ^{-h} \sum _{n = 0} ^{n^-} b _{n} \bar{z} ^{n}
	 \right )^{T}
\end{align*}
with $a _{n}, b _{n} \in \Cb$ and some integers $n ^{\pm}$.
Since the requirement $u \in \dom(D_a)$ in particular implies square integrability at infinity,
we obtain from the asymptotics~\eqref{eq:asymptot_h} of the potential function $h$
the conditions
 $n^+ - \frac{\Phi}{2\pi} < -1$
and 
 $n^- + \frac{\Phi}{2\pi} < -1$,
where 
 $\Phi = \Phi_0 +\sum_{k\leq N} \Phi'_j$.
From this we infer that there are $\floor{\frac{\Phi}{2\pi}}$ 
zero modes of spin up and 
 $\floor{-\frac{\Phi}{2\pi}}$ zero modes of spin down provided that 
 $|\Phi| > 2\pi$.
\end{proof}

\subsection{Proof for the  bounded region with holes} \label{sec:zero_modes_bdd_region}
In the case of the bounded domain the condition of the square integrability,
responsible for cutting off the infinite series in the final step 
in proof of  Thm.~\ref{thm:unbdd},
is substituted by the APS boundary condition~\eqref{eq:BC_hole_N+1} 
on the outer boundary 
 $\p \Omega _{out}$.
We denote by 
 $\mathcal{A}^{out}\subset M$ 
an open annulus 
whose outer boundary is $\p \Omega_{out}$ such that it satisfies 
$\mathcal{A}^{out}\cap \supp B = \emptyset$
and by
 $\widetilde{\Omega}$ the union $(\Omega_{out})^C \cup \mathcal{A}^{out}$
where $(\cdot)^C$ stands for the complement in $\Cb$.
To apply the boundary condition on 
 $\p \Omega_{out}$
we follow a similar process 
as in the case of checking the boundary conditions on the inner components of the boundary.
Let $\gamma (z_0^{out}, z)$ be a path connecting $z_0^{out} = R_{out}$ and a point 
$z\in \widetilde{\Omega}$. We define
\begin{align*}
 	G^+(z) 	&\coloneqq - \ii \int_{\gamma(z_0^{out}, z)} \vec{a} \dd \vec{s} +\int_{\gamma(z_0^{out}, z)} \frac{\Phi}{2\pi z'} \dd z'	\,, 	\\
 	G^-(z) 	&\coloneqq - \ii \int_{\gamma(z_0^{out}, z)} \vec{a} \dd \vec{s} -\int_{\gamma(z_0^{out}, z)} \frac{\Phi}{2\pi \bar{z}'} \dd \bar{z}' \,.
\end{align*}
By definition the restrictions to the boundary $\p \Omega _{out}$ satisfy
		\begin{align*}
 			G^{\pm}(z) \mid_{z \in \p \Omega _{out}} 
 				= -\ii \int_{\gamma_{out}} \vec{a} \dd \vec{s} + \ii \frac{\Phi}{2\pi} \phi \,,
		\end{align*}
		where $\gamma _{out} \subset \p \Omega _{out}$ connects the points 
		$z_{0}^{out} = R _{out}$ and $z\in \p \Omega _{out}$.
A direct adaptation of the proof of Lem.~\ref{lemma:path_indep_2}
shows that
$G^{\pm}(z)$ are independent of the choice of path $\gamma (z_0^{out}, z)$ 
contained in $\widetilde{\Omega}$.

We prove another key property of these functions.	
 \begin{lemma}\label{cor:exp}
 	Let us define
 		\begin{align*}
			F^{\pm}(z) \coloneqq \pm h(z) + G^{\pm}(z) \,.
		\end{align*}
 	Then
	$F^{+}(z)$ and $F^-(z)$
	 are analytic in $z$ and $\bar{z}$ respectively 
	 on $\widetilde{\Omega}$.
	Moreover, $F^{\pm}(z) \rightarrow const$ as $|z|\rightarrow \infty$.
\end{lemma}
\begin{proof}
	Similarly as in the proof of Prop.~\ref{cor:exp(F)},
	it can be shown that it holds
	\begin{align}
		\label{eq:F_plus}
 		F^+(z)  
 			& =  h(z_0^{out}) + \int_{\gamma(z_0^{out}, z)} \left(2\p_{z'} h + \frac{\Phi}{2\pi z'} \right) \dd z' \\
 		\nonumber
 		F^-(z)  
 			& =  -h(z_0^{out}) - \int_{\gamma(z_0^{out}, z)} \left(2\p_{\bar{z}'} h + \frac{\Phi}{2\pi \bar{z}'} \right) \dd \bar{z}' 	\,,
	\end{align} 
	where 
	 $\gamma(z_0^{out}, z) \subset \widetilde{\Omega} $ 
	is an arbitrary path connecting 
	 $z_0^{out}$ and 
	 $z$.
	Then $F ^{+}$ is analytic on 
	 $\widetilde{\Omega}$
	as 
	 $2\p_{z} h + \frac{\Phi}{2\pi z}$ 
	is analytic,
	and 
	 $F^-$ is anti-analytic as 
	 $2\p_{\bar{z}} h + \frac{\Phi}{2\pi \bar{z}} $
	is anti-analytic on that region (recall the Poisson equation $-\Delta h = B$
	 and Remark~\ref{rem:analyticity_of_path_indep_integral}).
	 
	Since we are further interested in the limit 
	 $|z| \rightarrow \infty$, let us assume that 
	 $|z|> R'$ for some $R'>2R _{out}$. 
	We will show that the absolute value of the integrand in ~\eqref{eq:F_plus} 
	decays like 
	 $|z| ^{-2}$ when 
	 $|z|$ tends to infinity.
	First for the singular parts of the magnetic field
	 $B_j = \Phi'_j \delta _{w_j}$ 
	we have for any 
	 $z \in \widetilde{\Omega}$
	\begin{align}\label{eq:estimate_sing}
 		2\p_z h_{[j]} + \frac{\Phi'_j}{2\pi z} 
 		=
 			 \frac{-\Phi'_j}{2\pi} \frac{w_j}{z(z-w_j)} \,,
	\end{align}
	with 
	 $h_{[j]} = \frac{-\Phi'_j}{2\pi} \log|z-w_j|$ 
	as before, in the proof of
	Prop.~\ref{cor:exp(F)}.
	In particular the absolute value of the right hand side is indeed bounded by a constant multiple of 
	 $|z|^{-2}$ for 
	 $|z|>R'$.
	For the bulk part of the magnetic field 
	 $B_0 \in C_0 ^{\infty}(M)$ 
	with the scalar potential 
	 $h_0(z) 
	 = 
	 	-\frac{1}{2\pi} \int_{M} B_0(z') \log|z-z'|
	 	\frac{\ii}{2} \dd z'\wedge \overline{\dd z'}$
	on $\widetilde{\Omega}$ we compute the derivative $\p_z h_0$ using the dominated convergence theorem
	similarly as in Remark~\ref{rem:a_is_bounded}.
	Then using the definition of the flux $\Phi_0$ we obtain the following estimate
	\begin{align} \label{eq:estimate_at_inf}
 		\left| 2\p_z h_0 + \frac{\Phi_0}{2\pi z} \right|
 			&= \left| -\frac{1}{2\pi} \int_{\Cb} \left( \frac{B_0(z')}{z-z'} - \frac{B_0(z')}{z} \right)  \frac{\ii}{2} \dd z'  \wedge \overline{\dd z'}  \right| \\
 			\nonumber
 			&\leq \frac{1}{2\pi}  \int_{\Cb}  \left| \frac{B_0(z')z'}{z(z-z')} \right| \frac{\ii}{2}\dd z'  \wedge \overline{\dd z'} 
 			\leq const |z| ^{-2} \,,
	\end{align}
	for all $|z|> R'$.
	In the last inequality we used that
	$\left| \frac{B_0(z')z'}{z(z-z')/|z|^2} \right| \leq 2 |B_0(z')z'|  \in  L^1(\Cb)$.
	Let us define
	\begin{align*}
		C_0 
 			& \coloneqq \int_0 ^{\infty} 
 				\left( 2\p_{z} h(z_0^{out}+t) 
 					+ \frac{\Phi}{2\pi (z_0^{out}+t)} \right) \dd t \,.		
	\end{align*}
	Then this is a well defined constant.
	Indeed, since an integral of analytic function along a bounded interval is bounded
	we have with use of \eqref{eq:estimate_at_inf} and \eqref{eq:estimate_sing}
	\begin{align*}
 			|C_0| 
 			\leq C_1 + \int _{R'}^{\infty} \frac{C_2}{t^2}  \dd t 
 			<\infty \,,
	\end{align*}
	with some constants $C_{1,2}>0$.
	Further by independence on the path 
	 $\gamma \subset \widetilde{\Omega}$
	(we can choose $\gamma$ along the real axis and then along the arc 
	corresponding to $|z|=const$) and by \eqref{eq:estimate_at_inf}, \eqref{eq:estimate_sing} we estimate
	\begin{align*}
		\left| \int _{\gamma} 2\p_{z'} h + \frac{\Phi}{2\pi z'} \dd z' -C_0 \right|
		 	\leq  const \left|-\int _{|z|}^{\infty} \frac{\dd t}{t^2}  +  
		 	 	 \frac{1}{|z|} \int_{0} ^{\arg (z)} \dd \phi \right| \,,
	\end{align*}
	which is arbitrarily small as $|z|  \rightarrow \infty$ and hence concludes the proof for 
	$F^+$. The proof of asymptotics for $F^-$ at infinity is analogous.
 \end{proof}

\begin{corollary} \label{cor:exp(F)_out}
	The exponentials of $F ^{\pm}$ have the following series on
	 $\widetilde{\Omega}$
	\begin{align*}
	 	\ee^{ F^+(z)} = \sum_{n\leq 0} d^+_n z^{n}	\quad \text{and} \quad
	 	\ee^{ F^-(z)} = \sum_{n\leq 0} d^-_n \bar{z}^{n} \,,
	\end{align*}	
	for some $d ^{\pm}_n \in \Cb$ with $d_0 ^{\pm} \neq 0$.
\end{corollary}
\begin{proof}
	By the previous lemma and by analyticity of $\exp(z)$ on $\Cb$ the function
	$\ee ^{ F ^{+}(w ^{-1})} $ is analytic and $\ee ^{ F ^{-}(w ^{-1})}$ is anti-analytic on 
	the interior of $\Cb\setminus \widetilde{\Omega}$ and converge 
	to a non-zero constant as $w \rightarrow 0$.
	This implies existence of the Taylor series  
	$\ee^{F^+(w ^{-1})} = \sum_{k\geq 0} d_k^{+} w^k $ and
	$\ee^{F^-(w ^{-1})} = \sum_{k\geq 0} d_k^{-} \overline{w}^k $
	with $d_0 ^{\pm}\neq 0$.
	Thus on the complement $\widetilde{\Omega}$ we have
	\begin{align*}
	 	\ee^{F^+(z)} = \sum_{n\leq 0} d_n^{+} z^{n}	\quad \text{and} \quad
	 	\ee^{F^-(z)} = \sum_{n\leq 0} d_n^{-} \bar{z}^{n} \,.
	\end{align*}	
\end{proof}

To apply the boundary condition on the outer boundary
we find the boundary values of a function that has the form of our zero modes when restricted to
 $\mathcal{A}^{out}\subset M$.
For conciseness we denote by $\check{H}_{out}$ the space $\check{H}(A|_{\p\Omega_{out}})$.

\begin{lemma} \label{le:traces_out}
	Let $u =(u^+, u^-)^T 
		= \sum_{n\geq 0} (\ee^h a_n z^n,\ee^{-h} b_n \overline{z}^n )^T
			\in L^2(\mathcal{A}^{out}, \Cb^2)$. 
	Then its trace on $\p \Omega^{out}$ is
		\begin{align} \label{eq:traces_out}
			\gamma_0 (u^+, u^-)^T 
		=
		 	\sum_{n\geq 0} \sum_{k\leq 0} R_{out}^{n+k}
		 	\begin{pmatrix}
			 	a_n d_k^+ \psi_{n+k}	
			 	\\
			 	b_n d_k^- \psi_{-(n+k)}			
		 	\end{pmatrix} 	
		\eqqcolon
			\begin{pmatrix}
		 		v_0^+ 	\\
		 		v_0^-
			\end{pmatrix} \,.
		\end{align}
		The vectors $\psi_{\ell}$ 
		were introduced in~\eqref{eq:eigenproblem_bdry_operator_out}
		and the coefficients $d_k^{\pm}$ are from 
		Cor.~\ref{cor:exp(F)_out}.
\end{lemma}
\begin{proof}
 	It is not difficult to see that Lem.~\ref{le:conv_in_max_dom} holds
 	also with $\mathcal{A}$ substituted by $\mathcal{A}^{out}$
 	 and $z-w_j$ substituted by $z$.
 	 Hence 
 	 $\gamma_0 (u)
 	 =
 	 	\lim_{Q\rightarrow \infty} 
 	 		\sum_{0\leq n\leq Q} 
 	 			R_{out}^n
 	 			(\ee^h a_n \ee^{\ii\phi n},\ee^{-h} b_n \ee^{-\ii\phi n} )^T$ in $\check{H}_{out} $.
 	 Mildly alternating the steps of the proof of Lem.~\ref{le:traces}
 	 we find the bound
 	 \begin{align*}
 	  	\|(\ee^{G^+}u^+, \ee^{G^-}u^-)\|_{L^2(\mathcal{A}^{out}, \Cb^2)}		  \geq const 
 	  	\|(v_0^+, v_0^-)^T\|_{\check{H}_{out}} \,,
 	 \end{align*}
 	showing that the sum in~\eqref{eq:traces_out}
 	converges in
 	 $\check{H}_{out}.$
 	The statement then follows from this convergence 
 	and by applying 
 	 Cor.~\ref{cor:exp(F)_out} to the expression
 	 \begin{align*}
 	  	\sum_{0\leq n\leq Q}
 	  		R_{out}^n
 	  		\left (
 	  			a_n \ee^{\ii\phi n} \ee^{-G^+(z)}
 	  				 \ee^{(h+G^+)(z)}, 
 	  			b_n \ee^{-\ii\phi n}\ee^{-G^-(z)}
 	  				\ee^{(h+G^-)(z)}
 	  		 \right )^T\Big|_{z \in \p\Omega_{out}} \,.
 	 \end{align*}
\end{proof}

The proof of the Aharonov--Casher result in the case of the bounded domain
 with holes and a circular outer boundary
now comes out along the lines of the proof of Prop.~\ref{cor:g_is_analytic}.

\begin{proof}[Proof of Thm.~\ref{thm:bdd}]
	Since the zero modes need to satisfy the APS boundary condition 
	on the inner components of the boundary,
	$ \p \Omega_j $, $j\leq N$, they have by 
	Prop.~\ref{cor:g_is_analytic} and Prop.~\ref{prop:zero_modes_form}  
	the form
	\begin{align*}
		\begin{pmatrix}
		 	u^+ 	\\
		 	u^-
		\end{pmatrix}
	=
		\sum_{n \geq 0} 
		\begin{pmatrix}
		 	\ee^h a_n z^n  	\\
	 		\ee^{-h} b_n \bar{z}^n
		\end{pmatrix}
	\end{align*}
	on the interior of $\Omega^{out}$.
	Using Lem.~\ref{le:traces_out} and the boundary 
	condition~\eqref{eq:BC_hole_N+1}
	we have
	\begin{align*}
	 	\gamma_0 	
	 		\begin{pmatrix}
	 		 	u^+ \\
	 		 	u^-
	 		\end{pmatrix}
		=
		 	\sum_{n\geq 0} \sum_{k\leq 0} R_{out}^{n+k}
		 	\begin{pmatrix}
			 	a_n d_k^+ \psi_{n+k}	
			 	\\
			 	b_n d_k^- \psi_{-(n+k)}			
		 	\end{pmatrix} 	
		=
		 	\begin{pmatrix}
			 	\sum_{\ell< \frac{\Phi}{2\pi} -\frac{1}{2} } 
			 		\beta_{\ell}^+ \psi_{\ell}	
			 	\\
			 	\sum_{\ell \geq \frac{\Phi}{2\pi} +\frac{1}{2}} 
			 		\beta_{\ell}^- \psi_{\ell}			
		 	\end{pmatrix} \,,
	\end{align*}
	with some $\beta_{\ell}^{\pm} \in \Cb$,
	which imposes the restrictions
	$a_n d_k^+ = 0$ if 
	$n+k \geq \frac{\Phi}{2\pi} -\frac{1}{2}$
	and 
	$b_n d_k^- = 0$ if
	$-(n+k) < \frac{\Phi}{2\pi} +\frac{1}{2}$.
	We recall that $d_0^{\pm} \neq 0$
	to deduce that there are
	 $ \floor{\frac{\Phi}{2\pi} - \frac{1}{2}} + 1 =  \floor{\frac{\Phi}{2\pi} + \frac{1}{2}}$ 
	 spin up 
	 and 
	$ \left\{ -\frac{\Phi}{2\pi} - \frac{1}{2} \right\} +1 =  \left\{\frac{-\Phi}{2\pi} + \frac{1}{2} \right\}$ 
	spin down zero modes.
	The symbol $\{y\}$ denotes the biggest integer smaller or equal to $y\geq 0$. 
	The proof is now concluded by noticing that the equality
	$\floor{y+\tfrac{1}{2}} = -\{-y+\tfrac{1}{2}\}$ holds
	 for any $y \leq \frac{1}{2}$. 
 \end{proof}

\section{ Aharonov--Casher on a sphere with holes} \label{sec:sphere}
In this section we will prove a version of the Aharonov--Casher theorem 
for the magnetic Dirac operator on a sphere with holes whose boundaries are
equipped  with APS boundary conditions.
This corresponds to our setup from Sec.~\ref{sec:problem_setting}
putting $\mathscr{M}=\Sb^2$.
In particular, let us consider the manifold $M  = \Sb^2 \setminus \cup_{k\leq N} \Omega_k$,
where $\cup_{k\leq N} \Omega_k$ is a union of mutually disjoint open discs on $\Sb^2$.
We again consider the magnetic field~\eqref{eq:tot_mag_field} on $M$
for which we additionally pose a requirement that the overall flux on the sphere sums to zero
\begin{align} \label{eq:mag_field_on_sphere}
 	\int _{\Sb^2} B_0 + B _{sing} = 0 \,.
\end{align}
To motivate the condition \eqref{eq:mag_field_on_sphere},
recall that the vector potential one-form $\alpha$ is globally defined
on $M$ 
and therefore the flux through the $N$-th hole is $\Phi_N = -\Psi$,
where $\Psi$ is the total flux minus $\Phi_N$.
This is so, since 
$\int _{\p \Omega_N} \alpha$
can be integrated either as $-\Psi$ or as $\Phi_N$
as $\p \Omega_N$ is boundary of both $\Omega_N$ and $\Omega_N^C$
which are both bounded regions.
Here $(\cdot)^C$ denotes the complement in $\Sb^2$.
We will consider a semi-total flux which we define as 
the bulk contribution $\Phi_0$ plus the normalised fluxes 
(\cf Sec.~\ref{sec:problem_setting}) through all the holes but one 
and we choose to omit the flux of the $N$-th hole
\begin{align*}
 	\widehat{\Phi} = \Phi_0 + \sum _{j\leq N-1} \Phi'_j \,,
 	\qquad 
 	\Phi_j' \in [-\pi, \pi) \,.
\end{align*}
The reasoning behind this comes from Lem.~\ref{le:gauge_invariance_sphere}
establishing the gauge invariance of this problem which we state later.
 It turns out that the problem of finding the zero modes is again gauge invariant
and one can gauge away integer multiples of $2\pi$ of the flux 
inside each of the holes apart from exactly one.
The number of zero modes then depends on the semi-total flux.
 Moreover the result does not depend on which hole was left out with non-normalised flux. 
 More precisely the following theorem holds.

\begin{theorem} \label{thm:sphere}
	Let $D$ be the Dirac operator on $M$ with magnetic field
	\eqref{eq:tot_mag_field} in the Aharonov--Casher gauge, that satisfies the condition
	\eqref{eq:mag_field_on_sphere}.
	 Then there are
	\begin{align*}
		\left| \floor{\frac{\widehat{\Phi}}{2\pi}+ \frac{1}{2}}\right|
	\end{align*}
	zero modes of the operator $D$ with the domain given by the APS boundary conditions.
	If $\widehat{\Phi} > 0$ then all the zero modes have spin up. 
	If $\widehat{\Phi} < 0$ then they have spin down.
\end{theorem}
The definition of the Dirac operator on a sphere is covered in  Appx.~\ref{ap:stereograph_proj} which also includes a proof of the statement that it is unitarily equivalent to the Dirac operator on a disc with holes with a conformal metric.
\begin{proof}
	We can rotate the 
	sphere so that the centre of the hole $\Omega_N$ becomes 
	the north pole $N'$.
	Then we perform a transformation $P$
	which is the stereographic projection from $N'$
	composed with a reflection (see also \eqref{eq:ster_proj}),
	to obtain a bounded region $P(M) \subset \Cb \simeq \Rb^2$ whose 
	all components of the boundary are circles. 
	This way we get the Dirac operator $D^W$ on the region $P(M)$ 
	with metric 
	\begin{align}
	\label{eq:conformal_metric_1}
	 	g^W= W^2 (\dd x^2 + \dd y^2)\,, 
	 	\quad \text{ where } \quad
	 	W= \bigg (
			1+ \frac{x^2 + y^2}{4}
			\bigg )^{-1} \,,
	\end{align}
	which is unitarily equivalent to the Dirac operator on $M$, by
	Cor.~\ref{cor:equivalence_sphere_and_bdd} in
	Appx.~\ref{ap:stereograph_proj}.		
	The statement is then a direct consequence of 
	Prop.~\ref{prop:conformal_case}
	proved below.
\end{proof}
  
\begin{remark}
\begin{enumerate}
\item
	Notice that in particular, there are no zero modes in the case 
	$\widehat{\Phi} \in (-\pi, \pi]$.
\item
	Let us point out that the number 
	$\left|
		\floor{\frac{\widehat{\Phi}}{2\pi}+ \frac{1}{2}}
	\right|$ 
	where 
	$\widehat{\Phi} =  \sum_{j\leq N-1} \Phi'_j$ does not depend 
	on the numbering of the holes.
	This is because we sum only over the normalised values 
	of the fluxes and due to the condition that the global flux is 
	zero, expressed by \eqref{eq:mag_field_on_sphere}.
	Hence if we fix an index $j_0\leq N-1$ and put 
	\begin{align*}
 		\Phi ^{I} = \Phi' _{j_0} + \Phi _{rest} \,,
	\end{align*}
	 where $\Phi _{rest} = \sum _{j\leq N-1, j\neq j _{0}} \Phi'_j$,
	we have by \eqref{eq:mag_field_on_sphere} the flux $-\Phi^I$ through the hole $\Omega_N$. 
	To normalise this value we note that for any $y \in \Rb$ it holds
	$y- \floor{y +\frac{1}{2}} \in (-\frac{1}{2}, \frac{1}{2}]$. Thus 
	\begin{align*}
 		\frac{\Phi' _{N}}{2\pi} 
 			= - \left(\frac{\Phi^I}{2\pi} - \floor{ \frac{\Phi^I}{2\pi} + \frac{1}{2}} \right) 
 			\in \left[ -\frac{1}{2}, \frac{1}{2} \right)\,,
	\end{align*}
	is the ($\frac{1}{2\pi}$ multiple of the) normalised flux through the $N$-th hole.
	The total flux $\Phi^{II} = \Phi _{rest}+ \Phi' _N$, \ie omitting the contribution from $j_0$, then satisfies
	\begin{align*}
 		\floor{ \frac{\Phi^{II}}{2\pi} + \frac{1}{2}}
 			= \floor{ \frac{\Phi _{rest}}{2\pi} + \frac{1}{2} - \frac{\Phi^I}{2\pi} + \floor{ \frac{\Phi^I}{2\pi} + \frac{1}{2}}  }
 			= \floor{ -\frac{\Phi' _{j_0}}{2\pi}  + \frac{1}{2} }  + \floor{ \frac{\Phi^I}{2\pi} + \frac{1}{2}}
 			= \floor{ \frac{\Phi^I}{2\pi} + \frac{1}{2}} \,,
	\end{align*}
	where in the last equality we used that $\frac{\Phi' _{j_0}}{2\pi} \in [ -\frac{1}{2}, \frac{1}{2})$.
\end{enumerate}
\end{remark}

To see that this result is a direct consequence of the bounded case we need to investigate the Dirac operator with the APS boundary condition under Möbius transform, 
which is a particular case of a conformal transform.

\subsection{The Dirac operator with APS boundary condition in the conformal metric $g^W$}

Let $M$ be a two-dimensional manifold 
with metric $g$
and let $E$ be a Spin$^c$ spinor bundle over $M$
with Clifford multiplication $\sigma$ and 
$Spin^c$ connection $\nabla$.
In \cite[Sec.~4]{ES01} the authors showed how
$\sigma$, $\nabla$ and 
the Levi-Civita connection $\nabla ^{LC}$
 are modified under a general conformal transformation
taking the metric $g$ to a metric $g^W = W^2 g$ for some $W:M   \rightarrow \Rb\setminus \{0\}$.
 We summarise their results in the following proposition.
 \begin{prop}\label{prop:conf_changes}
  	In the conformal metric $g^W = W^2 g$ we have
 	\begin{align*}
	\sigma^{W}(\mu) 
		&= W^{-1} \sigma(\mu) \,,	\\
	\nabla^{W}_{\mu^{\sharp}} u 
		&= \nabla_{\mu^{\sharp}} u +\frac{1}{4} W^{-1} [\sigma(\mu), \sigma (\dd W) ]u  \,, 	\\
	\nabla^{LC,W}_{\mu^{\sharp}}(\zeta)
		&= \nabla^{LC}_{\mu^{\sharp}}\zeta - W^{-1} \mu^{\sharp}(W)\zeta + W^{-1} (\zeta, \dd W) \mu - W^{-1} \zeta(\mu^{\sharp}) \dd W \,,	
\end{align*}
for any spinor $u$, vector field $\mu^{\sharp}$ and a one form $\zeta$.  
We denoted by $\mu $ the one-form dual to $\mu^{\sharp}$ with respect to the metric $g$.
 \end{prop}
 We point out that for any  $\zeta \in T^* M$ it holds
 $\sigma^W (W \zeta) = \sigma(\zeta)$
 and that if $\zeta$ is normalized in the metric $g$ then $W \zeta$ is normalized in the 
 conformal metric $g^W$.
As a consequence of Prop.~\ref{prop:conf_changes} we then obtain 
the relations of the Dirac operators and their boundary operators under a conformal transform.
In what follows, we will use the earlier introduced Notation~\ref{notation} on page~\pageref{notation}, 
 where the normalization refers to normalization in metric $g$.
\begin{corollary}\label{cor:A_W}
	Consider a two-dimensional manifold $M$ with the metric $g$ which is conformally 
	transformed to a manifold $M^W$ with metric $g^W = W^2 g$.
	The Dirac operators $D$ on $M$ and $D^W$ on $M^W$ and 
	their respective adapted boundary operators are related by 
	\begin{align*}
 		D^{W} 
			&=	W^{-3/2} D W^{1/2} 	\text{ and }	\\
		A^{W} 
			&= W^{-1}A	\,.
	\end{align*}
	In particular we see that the APS boundary condition is \emph{not} conformally invariant.
\end{corollary}
\begin{proof}
	The proof for $D^W$ is presented in \cite[Thm.~ 4.3]{ES01} so we show only the relation for $A^W$.
	Writing locally on the boundary $D = \sigma(\nu) \nabla_{\nu^{\sharp}} + \sigma(\xi) \nabla_{\xi^{\sharp}}$ 
	and using $\sigma(\nu)^2 = 1$ 
	we recall that by Def.~\ref{def:canonicalBO}
	the canonical boundary operator $A$ adapted to $D$ in the 
	metric $g$ reads
	\begin{align*}
	 	2A = \sigma(\nu) \sigma(\xi) \nabla_{\xi^{\sharp}} - \sigma(\xi) \nabla_{\xi^{\sharp}} \sigma(\nu)\,.
	\end{align*}
	Changing the metric from $g$ to $g^W = W^2g$ in this formula,
	Prop.~\ref{prop:conf_changes}
	further gives
	\begin{align*}
	 	2A^W &= \sigma(\nu) \sigma(\xi) W ^{-1} \nabla^W_{\xi^{\sharp}} - \sigma(\xi) W ^{-1} \nabla^W_{\xi^{\sharp}} \sigma(\nu)		\\
	 		&= W ^{-1}\Big( \sigma(\nu) \sigma(\xi) \big(\nabla_{\xi^{\sharp}} +\frac{1}{4} W^{-1} [\sigma(\xi), \sigma (\dd W) ] \big) 	
	 		- \sigma(\xi) \big(\nabla_{\xi^{\sharp}}  +\frac{1}{4} W^{-1} [\sigma(\xi), \sigma (\dd W) ]\big) \sigma(\nu)\Big)	\\
	 		&= W ^{-1} \left(\sigma(\nu) \sigma(\xi) \nabla_{\xi^{\sharp}} - \sigma(\xi) \nabla_{\xi^{\sharp}} \sigma(\nu)\right) + \frac{W ^{-2}}{4} R
	 		= W ^{-1} 2A  +\frac{W ^{-2}}{4} R \,,
	\end{align*}
	where
	\begin{align*}
	 	R &\coloneqq 
	 		\sigma(\nu) \sigma(\xi) [\sigma(\xi), \sigma (\dd W) ] -  \sigma(\xi) [\sigma(\xi), \sigma (\dd W) ] \sigma(\nu)	
	 		=  
	 		-\sigma(\xi) \{ [\sigma(\xi), \sigma (\dd W) ], \sigma(\nu)\} \,.
	\end{align*}
	Since the pair $(\nu, \xi)$ forms a local orthonormal basis of the one forms
	we can write $\dd W = (\dd W, \xi) \xi +(\dd W, \nu) \nu$ to obtain
	 \begin{align*}
	 	[\sigma(\xi), \sigma(\dd W)] 
	 	= (\dd W, \nu) \left(\{\sigma(\xi), \sigma(\nu) \} -2\sigma(\nu)\sigma(\xi) \right)
	 	= -2(\dd W, \nu) \sigma(\nu)\sigma(\xi)  \,.
	\end{align*}
	Therefore using the anti-commutation identity $\{EF,G\} =E\{F,G\} - [E,G]F$
	holding	for any operators $E,F,G$,
	we infer $R=0$,
	which concludes the  proof of $A^W = W ^{-1}A$.
\end{proof}

Let us restrict to the specific case of the Dirac operator on $P(M)$
with $M = \Sb^2\setminus \cup_{j\leq N} \Omega_j$
and $P$ the stereographic projection composed with a reflection 
defined in Appx.~\ref{ap:stereograph_proj}.
Note, that $P(M)$ is a conformal transformation of $\Omega_{out} \setminus \cup_{j\leq N-1} \Omega_j$
and the new metric $g^W$ is given by~\eqref{eq:conformal_metric_1}.
As in the case of the standard metric
we can use the arguments for gauge invariance
from Lem.~\ref{le:gauge_invariance} for the holes $P(\Omega_j)$, $j\leq N-1$
to find the following.
\begin{lemma}\label{le:gauge_invariance_sphere}
	Let $a $ and $\widetilde{a}$ be two magnetic vector potentials whose fluxes differ by an integer
	 multiple $m_j$ of $2\pi$ on the inner hole $\Omega_j$, for all $j\leq N-1$.
	 Then we have the unitary equivalence between the Dirac operators
	 on $P(M)$ in the metric $g^W$ with APS boundary condition, 
	 corresponding to the 
	 magnetic fields $a$ and $\widetilde{a}$
	 \begin{align*}
 		\mathscr{U}^{\ast} D^W_a \mathscr{U} = D^W_{\widetilde{a}} 	\,, 
	\end{align*}
	with the unitary
	 \begin{align*}
	 	\mathscr{U}	&:  L^2(\Cb, g^W; \Cb^2)  \rightarrow  L^2(\Cb, g^W; \Cb^2) 	\\
 		\mathscr{U}	&: u \mapsto \exp \bigg[\ii \int_{\gamma} (\vec{a} -\widetilde{\vec{a}}) \dd \vec{s} \bigg] u \,,
	\end{align*} 
	where $\gamma\subset P(M)$ is a curve connecting a fixed point $z_0\in P(M)$ and the point $z$.
\end{lemma}
We recall from the proof of Lem.~\ref{le:gauge_invariance}, 
that replacing the starting point $z_0$ of $\gamma$ 
by a different point $z_1\in P(M)$ amounts merely to a multiplication by the constant
$K =  \exp \bigg[\ii \int_{z_0}^{z_1} (\vec{a} -\widetilde{\vec{a}}) \dd \vec{s} \bigg] $ satisfying $\overline{K} = K ^{-1}$.

\begin{proof}
	Taking the commutativity of $\mathscr{U}$ and $W$ into account, 
	the statement follows directly from Lem.~\ref{le:gauge_invariance} and Cor.~\ref{cor:A_W}.
\end{proof}
By this lemma we can without loss of generality 
 work with the normalized fluxes $\Phi'_j \in [-\pi, \pi)$
 inside the holes $\Omega_j\subset \Sb^2$ for $j\leq N-1$ 
 as well as assume that the magnetic field inside the holes is modelled  
 by such a normalized flux multiple of the Dirac delta function $\delta_{w_j}$, 
 at the centre $w_j$ of $\Omega_j$.
Harvesting all this preparation we are able to find the zero modes 
of the conformal Dirac operator on $\Cb$ and prove the following 
proposition whose immediate consequence is Thm.~\ref{thm:sphere}.
\begin{prop}\label{prop:conformal_case}
 	The zero modes of the Dirac operator $D^W$ on $P(M)$ 
 	in the metric $g^W$ and magnetic field satisfying the condition~\ref{eq:mag_field_on_sphere}
 	in the Aharonov--Casher gauge 
 	with the APS boundary condition are of the form
 	\begin{align*}
 		\begin{pmatrix}
 			u^+ 	\\
 			0
		\end{pmatrix} \,, \quad
 		u^+(z) 
 		&=
 			 W ^{-1/2} (z) \ee ^{h (z)} \sum _{0\leq n < \frac{\widehat{\Phi}}{2\pi} - \frac{1}{2}} a_n z^n	\,,	\\
 		\begin{pmatrix}
 			0	\\
 			u^-
		\end{pmatrix} \,, \quad
 		u^-(z)  
 		&=
 			 W ^{-1/2}(z) \ee ^{-h (z)} \sum _{0\leq n \leq -\frac{\widehat{\Phi}}{2\pi} - \frac{1}{2}} b_n \overline{z}^n
	\end{align*}
	for some coefficients $a_n, b_n \in \Cb$.	
\end{prop}
\begin{proof}
	Consider a zero mode $u \in \ker(D^W)$. 
	Then by Cor.~\ref{cor:A_W} we know that for $v(z)= W(z)^{1/2}u$  it holds $Dv(z) = 0$ on $P(M)$
	with $D = W ^{3/2}D^W W ^{-1/2} $ being the Dirac operator
	on $P(M)$ in the standard metric on $\Cb$.
	We choose coordinates $\tilde{z}$ on 
	 $P(\Sb^2\setminus \{N'\})$ with origin at 
	 $P\big( (0,0,-1)^T\big)$ and mark with tilde
	 functions on $P(M)$ expressed in these coordinates.
	Let us fix an arbitrary index $j\leq N-1$.
 	We write similarly 
 	 $f_j(z_j)$ for a function 
 	 $f$ on $P(M)$  
 	in the coordinates $z_j$ obtained by the Möbius transform 
 	 $Y _{t_j}:\tilde{z} \mapsto z_j$
 	(see Appx.~\ref{ap:mobius_transform} and Lem.~\ref{le:def_Y})
 	with $t_j$ being the antipodal point of the centre 
 	$w_j$ of the hole $\Omega_j \subset \Sb^2$.
 	An important observation is that
	 $W_j(z_j)$ is a positive constant on $(Y _{t_j} \circ P) (\p \Omega_j)$
	and therefore 
	 $u$ satisfies the APS boundary condition on 
	 $(Y _{t_j} \circ P) (\p \Omega_j)$ for $D^W$ 
	if and only if $v$ satisfies 
	the boundary condition \eqref{eq:BC_hole_j} on
	 $(Y _{t_j} \circ P) (\p \Omega_j)$.
 	By Prop.~\ref{prop:zero_modes_form} the spin up component $v^+$ takes the form
	\begin{align}\label{eq:sphere_spin_up}
	 	 v _{j}^+(z) = \ee ^{h _j(z)} g _{j}^+(z) \,,		\quad j \leq N -1 \,,
	\end{align}
	where $g_j^+(z)$ is analytic on $(Y _{t_j} \circ P) (M)$ and   
	 can be analytically extended to the hole $(Y _{t_j} \circ P) (\Omega_j)$
	 by Prop.~\ref{cor:g_is_analytic}. 	  
	 In Appx.~\ref{ap:mobius_transform} we argue that under the change of coordinates given by the Möbius transform  
	 \begin{align*}
 		Y _{t_j}: \tilde{z} \mapsto z_j = \frac{a \tilde{z} + b}{c \tilde{z} +d} \,,
	\end{align*}
	for some $a, b, c, d$ 
	complex numbers dependent on $t_j$
	(the number $a$ here should not be confused with the 
	vector potential $a$), 
	 the spinor $u$ needs to satisfy the relation~\eqref{eq:patching},
	 and therefore
	\begin{align*}
	 	W_j ^{-1/2}(Y_{t_j}(\tilde{z})) v_j^+ (Y_{t_j}(\tilde{z})) = \tilde{W} ^{-1/2}(\tilde{z}) \tilde{\mathscr{G}} (\tilde{z}) \tilde{v}^+(\tilde{z})\,,	
	 	\quad \mathscr{G}(z) = 
	 		\frac{cz+d}{|cz+d|}	\,,
	\end{align*}
	for all $j\leq N-1$.
	Employing \eqref{eq:W_ratio} 
	and \eqref{eq:sphere_spin_up} this now leads to analyticity of $\tilde{g}^+(\tilde{z})$ on $P(\Omega_j)$ as
	\begin{align*}
	 	g^+_j (Y_{t_j}(\tilde{z})) 
	 		= \ee ^{\tilde{h} (\tilde{z})-h_j(Y_{t_j}(\tilde{z}))} |c\tilde{z} +d| \frac{c\tilde{z}+d}{|c \tilde{z}+d|} \tilde{g}^+(\tilde{z})
	 		=(c \tilde{z} +d) \tilde{g}^+ (\tilde{z}) \,,
	\end{align*}
	where we used that the functions $h_j(Y_{t_j}(\tilde{z}))$ and $\tilde{h}(\tilde{z})$ are in fact the same function
	$h$ expressed in different sets of coordinates.
	Hence using that $(c \tilde{z} +d) ^{-1}$ is analytic on $P(\Omega_j)$
	\footnote{Note that the point $\tilde{z} = -d/c \notin P(\Omega_j)$ 
	is in fact the image of the antipodal point $t_j$ of $w_j$ under the mapping $P$.}
	and that $j\leq N-1$ was arbitrary
	we conclude that $g^+$ is analytic on $P(\Sb^2 \setminus \Omega_N)$.
	Similarly as above thanks to $W(\tilde{z}) = const>0$
	on $P(\p \Omega_N)$ the boundary conditions on $P(\p \Omega_N)$ 
	for 
	 $u\in \dom(D^W)$ and $v\in \dom(D)$ coincide
	(see Cor.~\ref{cor:A_W}).
	Therefore we may apply the same steps as in the proof of
	Thm.~\ref{thm:bdd} and obtain
	\begin{align*}
 		u^+(z) = W ^{-1/2}(z) \ee ^{h(z)} \sum _{n < \frac{\widehat{\Phi}}{2\pi} -\frac{1}{2}} a_n z^n 
	\end{align*}
	on $P(M)$.
	The form of the modes $u^-$ on $P(M)$ is shown by an adaptation 
	of the previous to be
	\begin{align*}
 		u^-(z) = W ^{-1/2}(z) \ee ^{-h(z)} \sum _{n \leq -\frac{\widehat{\Phi}}{2\pi} -\frac{1}{2}} b_n \bar{z}^n \,.
	\end{align*}	
	Here both $a_n$ and $b_n$ are some complex coefficients.
 \end{proof}

\section{Relation to the index theorem}\label{sec:index}
Here we assume for a moment that the dimension 
 $n$ of the base manifold 
 $M$ is even (not necessarily two).
In that case a Spin$^c$ spinor bundle $E$
with Clifford multiplication $\sigma$
and a $Spin^c$ connection can be defined, provided 
that a certain topological
condition\footnote{
	For further details on this condition see~\cite[Thm.~D.2]{LM89}
	} 
is imposed on $M$.
We refer to \eg \cite[Appx.~D]{LM89} or 
\cite[Sec.~10.8]{TayII})
for the precise definitions of the above terms in dimension $n>2$, 
for $n=2$ recall Sec.~\ref{sec:geometry}.
Due to the Clifford relations
the chirality operator then anti-commutes with $\sigma(\zeta)$
for all $\zeta \in T ^{\ast}M \subset Cl(\Rb^n)$ 
with $Cl(\Rb^n)$ denoting the Clifford algebra on $\Rb^n$
and induces thus
 a $\Zb_2$ grading of the bundle $E$. 
This means that
we can write $E = E_+ \oplus E_-$ where $E _{\pm}$ are the 
$\pm 1$ eigenspaces or the chirality operator.
If $D$ is the Dirac operator on $E$, it can be then written in the following  form
\begin{align*}
 	D = \begin{pmatrix}
 	 		0	&	D_-	\\
 	 		D_+	&	0
 		\end{pmatrix} \,,
\end{align*}
where $D_{\pm}: \Gamma (E_{\pm}) \rightarrow \Gamma(E_{\mp})$ are mutual formal adjoints.
We remark, that 
the Dirac operator is defined by the same formula as in two dimensions 
with the difference that now the index $j$ in 
Def.~\ref{def:Dirac_op} runs up to $n$.
We wish to introduce the quantity index, which is well-defined for Fredholm operators. Therefore in the following we assume that $D$ is 
Fredholm. That is for example true when $M$ is a compact manifold and 
$D$ satisfies the APS boundary condition on $\p M$, \cf
\cite[Def.~5.1, Ex.~5.2, Thm.~5.3]{BB}
\footnote{\cite{BB} shows that it is true even for all $D$-elliptic boundary conditions, see \cite[Def.~4.7]{BB} for the definition.}.
\begin{definition}
 	We define the \demph{analytical index} (or \demph{index}) of the Dirac operator $D$ by
 	\begin{align}\label{eq:analytic_ind}
 		\ind(D) = \dim \ker(D_+) - \dim\ker(D_-) \,.
	\end{align}
\end{definition}
Atiyah and Singer showed in \cite{AS63} that if the manifold $M$ is compact and has no boundary, 
then the analytical index is equal to the topological index 
\begin{align}\label{eq:AS}
 	\ind (D) = \int_M AS \,.
\end{align}
The integrand $AS$ depends both on the Riemannian curvature $R_M$ of $M$ and
 the magnetic 
 two-form\footnote{In a general dimension the magnetic two-form is 
 the trace $2^{-n/2} \mathrm{Tr} (\ii R)$
 of the $\mathrm{End}(E)$-valued curvature $R(X,Y) = \nabla_X \nabla_Y - \nabla_Y\nabla_X - \nabla_{[X,Y]}$ of the $Spin^c$ connection $\nabla$.
 Here $X,Y$ are arbitrary vector fields on $E$.}
 $\beta$ on the bundle $E$.
 For flat
manifolds, \ie $R_M = 0$, it corresponds to the Chern character of the bundle
  $AS = \mathrm{Ch}(E)_{[n]}	= \left(\exp \frac{ \beta}{2\pi }\right)_{[n]}$,
where the subscript $[n]$ refers to the  $n$-th degree part of the form.
The expression $\exp$ is to be understood as the series expansions.

The index theorem was extended to manifolds with boundary in~\cite{APS1},
where Atiyah, Patodi and Singer proved the formula for the index
	assuming, that $M$ has a product structure near the boundary.
	Neglecting this assumption one obtains an additional boundary term
	that vanishes in the case of a product structure.
	The extended formula was proven by Grubb in \cite[Cor.~5.3]{Gru92}.
	More explicit expression of the boundary term was given by Gilkey in \cite{Gi93}.
	From Gilkey's formula it follows that in case of the APS boundary condition given by the canonical boundary operator (\cf Def.~\ref{def:canonicalBO})
	the above mentioned additional boundary term vanishes.
	In particular in our two-dimensional case we obtain by Thm.~8.4.d and Thm.~1.4 in \cite{Gi93}
	\begin{align} \label{eq:ind_form}
		\ind (D) = \int_M AS - \frac{1}{2}(\eta([{A}] _{11})+ \dim \, \ker[{A}]_{11}) \,,
	\end{align}
	and  $[{A}] _{11}$ is its top left component of the canonical boundary operator.
	
	We will consider 
	the Dirac operator $D_a$ with magnetic field
	\eqref{eq:tot_mag_field}. Recall that 
	$A$ was computed in Sec.~\ref{sec:problem_setting}.
	The first term in the integral
	is the bulk contribution
	as in \eqref{eq:AS}.
	Since in our case $M$ is flat,
	and, since we are in two dimensions,	
	 we have
	  $\int_M AS = \int_M \frac{\beta}{2\pi} 
	 	= \frac{\Phi_0}{2\pi}$.	
	The $\eta$-invariant $\eta(A)$ is defined 
	as the analytic extension 
	of the function
	\begin{align*}
 		\eta_s(A) = \sum_{\lambda \in spec(A)\setminus \{0\} } |\lambda|^{-s} sgn(\lambda) \,,
	\end{align*}
	at the value $s =0$ and 
	is well defined for Dirac operators as was shown in~\cite{APS1}.
	The sum runs over the non-zero eigenvalues of the boundary operator $A$.
	For the simple case $T = -\ii \p_t -c$, $c \in \Rb$,
	on the first Sobolev space $H^1([0, 2\pi])$
	with periodic boundary condition,
	it is shown in a greater detail for example in \cite[Appx.~D]{Thesis}
	that the analytic continuation yields
	\begin{align*}
 		\eta( -\ii \p_t -c) =
 			\begin{cases}
				-1+ 2\langle c \rangle 	& \text{ if } c \in \Rb\setminus \Zb 		\\
				0				& \text{ if } c \in \Zb
			\end{cases} \,,
	\end{align*}
	where $\langle c \rangle$ is the unique number $\widetilde{c} \in (0, 1)$ such that $c-\widetilde{c}\in \Zb$.
	Note that the eta-invariant $\eta(T) $ depends only on the eigenvalues 
	of $T$ and hence the formula for  
	$\eta(T)$ yields directly a result for the $\eta$-invariant of any operator whose spectrum is of the form $\{n+c\}_{n\in \Zb}$.
	Solving the eigenvalue problem for the top left component 
	of the boundary operator $[{A}_{11}]$
	restricted to the inner and outer components of the boundary 
	gives the spectra
	\begin{align*}
 		spec\big([A|_{\p\Omega_j}]_{11} \big)
 			&= \{-R_j ^{-1}
 				\bigg(n - \frac{\Phi'_j}{2\pi} +\frac{1}{2} 
 				 \bigg) \mid n \in \Zb\} 	\\
 		spec\big([A |_{\p\Omega_{out}}] _{11} \big)	
 			&=\{R _{out} ^{-1}
 				\bigg(n - \frac{\Phi}{2\pi} +\frac{1}{2}  
 				 \bigg) \mid n \in \Zb\} \,.
	\end{align*}
	Employing then the property $\eta(c L) = \mathrm{sgn}(c)\eta(L)$ for a constant $c$ and an elliptic operator $L$
	we deduce form~\eqref{eq:eigenproblem_bdry_operator} and 
	\eqref{eq:eigenproblem_bdry_operator_out}
	\begin{align*}
		\eta\big([A|_{\p\Omega_j}]_{11}\big)
		=
			1-2 \bigg\langle \frac{\Phi'_j}{2\pi} - \frac{1}{2}  \bigg\rangle 
		\quad \text{and} \quad
		\eta\big([A|_{\p\Omega_{out}}]_{11}\big) 	
		=
			-1+2 \bigg\langle \frac{\Phi}{2\pi} - \frac{1}{2}  \bigg\rangle
	\end{align*}
	for all $j\leq N$.
	Let us denote by $I_1$ the set of indices $j$ such that $1=\dim\,\ker([A|_{\p\Omega_j}]_{11}) \in \{0, 1\}$, 
	by $|I_1|$ the number of elements in $I_1$
	and let $I_0 = \dim([A|_{\p\Omega_{out}}]_{11}) \in \{0,1\} $.
	We make the following observations
	\begin{enumerate}
		\item By definition of the normalized fluxes (recall Sec.~\ref{sec:problem_setting}) if $j \notin I_1$ we have 
				$\bigg\langle \frac{\Phi'_j}{2\pi} -\frac{1}{2} \bigg \rangle
 				=
 					\frac{\Phi'_j}{2\pi} +\frac{1}{2}$.
		\item For $j \in I_1$ it holds
			 $\frac{\Phi'_j}{2\pi} -\frac{1}{2} = -1 $
			 and thus,
			 $\sum _{j\in I_1} \left( \frac{\Phi'_j}{2\pi} \right)
 				=
 					 -\frac{|I_1|}{2}$.
		\item 
			$\eta([A|_{\p \Omega_{out}}]_{11}) + I_0 = 
 				\begin{cases}
					1 	
						& \text{ if } I_0 = 1 	\\
					-1+2 \left\langle \frac{\Phi}{2\pi} - \frac{1}{2}  \right\rangle 
						& \text{ if } I_0 = 0
				\end{cases} \,.$
\end{enumerate}
	Omitting the outer boundary contribution in the index formula~\eqref{eq:ind_form} for now, 
	we arrive at the expression
	\begin{align*}
		\int_M AS 
			- \frac{1}{2}\sum _{j\leq N} (\eta([A|_{\p\Omega_j}] _{11})
			+ \dim \, \ker[A|_{\p\Omega_j}]_{11}) 
		&= 
			\frac{\Phi_0}{2\pi} -\frac{1}{2}\sum _{j\notin I_1} 
			-2 \left(\frac{\Phi'_j}{2\pi} \right) 
			-\frac{|I_1|}{2}	\\
		&= 
			\frac{\Phi_0}{2\pi} 
			+ \sum _{j\leq N} \left(\frac{\Phi'_j}{2\pi} \right)
		= 
			\frac{\Phi}{2\pi} \,.
	\end{align*}
	Finally for the index of the Dirac operator
	$D_a$ with magnetic field
	\eqref{eq:tot_mag_field} 
	in the gauge~\eqref{eq:h_relation},~\eqref{eq:potential} and
	with 
	domain~\eqref{eq:APS_domain} we have
	\begin{align*}
 		\ind (D_a) &= \frac{\Phi}{2\pi}- 
 			\left.
 			\begin{cases}
 				\frac{1}{2} 		& \text{ if } I_0 = 1	\\
 				-\frac{1}{2}+ \left\langle \frac{\Phi}{2\pi} - \frac{1}{2} \right\rangle & \text{ if } I_0 = 0
			\end{cases}	\right\} 	
			=\floor{ \frac{\Phi}{2\pi} + \frac{1}{2}} \,,
	\end{align*}
	where in the last equality we used that $\frac{\Phi}{2\pi} + \frac{1}{2} \in \Zb$ if $I_0 = 1$.
	Note that this formula is in agreement with our result 
	Thm.~\ref{thm:bdd}, by which we immediately infer:
\begin{corollary}\label{cor:our_index}
	Under the assumptions of Thm.~\ref{thm:bdd}  
	we obtain the index for $D_a$ (defined by~\eqref{eq:analytic_ind}),
	\begin{align*}
		\ind(D_a) = \floor{\frac{\Phi}{2\pi} + \frac{1}{2}} \,.
	\end{align*}
\end{corollary}

\section{Conclusion}
We showed a version of the Aharonov--Casher theorem on some
two-dimensional manifolds with boundary. 
In particular our manifolds are 
a plane with holes (Thm.~\ref{thm:unbdd}), 
a disc with holes (Thm.~\ref{thm:bdd}) and 
a sphere with holes (Thm.~\ref{thm:sphere}).
We consider the APS boundary condition on the boundary and show that
the number of zero modes depends only on the sum of the
flux corresponding to the smooth
magnetic field on the manifold and the rational part of the
fluxes through the holes.
In particular our results imply the index theorem for these special 
choices of the manifolds. Moreover since the index is a topological
invariant, the index theorem is implied also for arbitrarily shaped holes.
To prove the Aharonov--Casher theorem, \ie , 
that all zero modes have a definite chirality, is, 
for such domains still an open problem.

\appendix
\section{The Dirac operator under the stereographic projection} \label{ap:stereograph_proj}
	For conciseness we will write in the following $\widetilde{M} =\Sb^2 \setminus \{N'\}$ for the sphere without the north pole $N' = (0, 0, 1)^T \in \Sb^2$ and
$M = \Sb^2\setminus \cup_{j\leq N} \Omega_j$ for the sphere without the holes
$\Omega_j \subset \Sb^2$, $j\leq N$.
	It is convenient to map the Dirac operator on the sphere
	to the plane by the stereographic projection.
	Here we will argue that due to this mapping we can perform the analysis for finding the zero 
	modes of the Dirac operator on $\Sb^2$
	by investigating the problem on $\Cb$ with a metric that is conformal 
	to the standard metric on $\Cb$.
	We will denote by $P: \widetilde{M} \rightarrow \Cb$
	the stereographic projection from the north pole 
	composed with reflection across the $x$ axis. In particular a point
	\begin{align*} 
			\omega =
 			\begin{pmatrix}
  				\cos \phi \sin \theta\\
  				\sin \phi \sin \theta 	\\
  				\cos \theta
			\end{pmatrix} \,, \quad
			\theta \in (0, \pi]\,, \phi\in (0, 2\pi] \,,
		\end{align*}
		is mapped by $P$ to the point $P(\omega) = 2 \cot \frac{\theta}{2} \ee ^{-\ii \phi} \in \Cb$, \ie
		\begin{align} \label{eq:ster_proj}
 			x \coloneqq (P(\omega))_x = 2 \cot (\theta/2) \cos\phi 	\,,\quad 
 			y \coloneqq (P(\omega))_y = -2 \cot (\theta/2) \sin \phi \,.
		\end{align}
	\begin{lemma} \label{le:sgp_isometry}
		The tangent map $P _{\ast}: (T\widetilde{M}, g^{\Sb^2})  \rightarrow (T\Rb^2, g^W)$,
		where $g^{\Sb^2}$ is the standard metric on $\Sb^2$ 
		and 
		\begin{align}\label{eq:W}
		 	g^W = W^2 (\dd x^2 + \dd y^2)\,,
		 	\quad W= \left(1+ \frac{x^2 + y^2}{4} \right)^{-1} \,,
		\end{align}
		is an isometry.
	\end{lemma}
	\begin{proof}
		Using the definition of the push-forward map $(\cdot)_{\ast}$
		the statement follows by a direct computation from~\eqref{eq:ster_proj}.		
	\end{proof}
	We further obtain a unitary between the square integrable functions over
	$\Cb$ with the metric $g^W$ and the square integrable functions over the pre-image $P ^{-1}(\Cb)$
	with the standard metric on $\Sb^2$.

\begin{lemma}\label{le:P_ast}
	The pullback of the stereographic projection composed with reflection across the $x$ axis \\
	 $P ^{\ast}: L^2(\Cb, g ^W; \Cb^2)  \rightarrow L^2(P ^{-1}(\Cb), g ^{\Sb^2}; \Cb^2) $ acting as
 	$(P ^{\ast} u)(\omega) \coloneqq u(P(\omega)) $ is a unitary operator.
\end{lemma}
\begin{proof}
	Finding the differentials $\dd x$ and $\dd y$ from~\eqref{eq:ster_proj}
	one easily verifies that the volume form changes as
	\begin{align*}
	 	\dd \theta \wedge \sin \theta \dd \phi 
	 		&= \sin^4 (\theta/2) \dd x\wedge \dd y 
	 		  =  \left( 1+ \frac{x^2 + y^2}{4} \right) ^{-2} \dd x \wedge \dd y\,.	
	\end{align*}
	With the notation $(\cdot,\cdot)_{\Sb^2}$ for the inner product on $L^2(\Sb^2, g^{\Sb^2}; \Cb^2)$
	and $(\cdot,\cdot)_W$ for the inner product on $L^2(\Cb, g^W; \Cb^2)$
	 we then obtain
	\begin{align*}
	 	 (P ^{\ast} f_1, P ^{\ast} f_2) _{\Sb^2}
	 		&=\int_0 ^{\pi} \int_0 ^{2\pi} f_1 \circ P(\theta, \phi) \, \overline{f_2 \circ P(\theta, \phi)} \dd \theta \wedge \sin \theta \dd \phi 	\\
	 		&=\int_{\Rb^2} f_1 (x,y) \overline{f_2(x,y)} \left( 1+ \frac{x^2 + y^2}{4} \right)^{-2}  \dd x\wedge \dd y
	 		= (f_1, f_2)_W \,,
	\end{align*}
	 for any $f _{1,2} \in L^2(\Cb, g^W; \Cb^2)$ and for
	$W$ given by~\eqref{eq:W}. 
\end{proof}

\begin{definition}
	We define the Spin$^c$ spinor bundle over $\widetilde{M}$ 
	as the pullback of the Spin$^c$ spinor bundle 
	 $\mathcal{S}$ over $\Cb \sim \Rb^2$ 
	by the stereographic projection composed with reflection $P$
	\begin{align*}
	 	P^{\ast} \mathcal{S}
	 		= \{ (\omega, u)  \in \Sb^2 \times \mathcal{S} \mid \pi (u) = P(\omega) \} \,,
	\end{align*}
	where $\pi$ is the bundle projection of $\mathcal{S} $.
	We have as in Lem.~\ref{le:P_ast} the map 
	$P ^{\ast}: \Gamma(\Rb^2, \mathcal{S})  \rightarrow \Gamma(\widetilde{M}, P ^{\ast} \mathcal{S})$
	given by 
	$(P ^{\ast} u)(\omega) = (u \circ  P)(\omega)$.
	The corresponding Clifford multiplication and the Clifford connection 
	on such bundle are given by
	\begin{equation} 
	\label{eq:pullback_of_composition}
		\begin{aligned}
		 	\sigma^{\widetilde{M}}(P 	^{\ast}\zeta) P^{\ast}
		 		&\coloneqq P ^{\ast} \sigma^W(\zeta) \,,		
		 		& \zeta\in T ^{\ast}\Rb^2\\
		 	\nabla^{\widetilde{M}}_{X} P^{\ast}
		 		&\coloneqq P ^{\ast} \nabla^W_{P _{\ast}X} \,,  	
		 		& X \in T\widetilde{M} \,,
		\end{aligned} 
	\end{equation}
	where $\sigma^W$ and $\nabla^W$ refer to the Clifford multiplication 
	and Clifford connection on $\mathcal{S}$.
 \end{definition}
 Now we are ready to state a corollary which will 
 reduce our analysis of the Dirac operator on the 
 the sphere with holes
 to the investigation of the corresponding Dirac operator on 
 a disc with holes
 in a metric conformal to the standard metric on $\Rb^2\simeq \Cb$.
 
\begin{corollary}\label{cor:equivalence_sphere_and_bdd}
	The Dirac operator $D^M$ on $M$ is unitarily equivalent to the Dirac operator $D^W$
	on $P(M) \subset (\Cb, g^W)$,
	\begin{align*}
	 	 D ^{M} P ^{\ast} = P ^{\ast} D^W \,.
	\end{align*}
\end{corollary}
\begin{proof}
	We denote by $s^j$, $j=1,2$ an orthonormal (in $g^{\Sb^2}$) basis on $T ^{\ast}M$,
	by $s_j$ the dual basis and 
	by $e^j$ its counterpart on $T ^{\ast}\Rb^2$ such that
	 $P ^{\ast}e^j = s^j$. 
	 Note, that by Lem.~\ref{le:sgp_isometry} 
	 the last relation defines $e^j$ that form an orthonormal frame on 
	 $T ^{\ast}\Rb^2$ in the metric $g^W$. 
	Using the definitions~\eqref{eq:pullback_of_composition} we obtain for any section $u$ on $\Rb^2$
	\begin{align*}
 		D ^{M} P ^{\ast} u
 			= \sum _{j\leq 2} \sigma^{M}(s^j) \nabla^{M}_{s_j}
 				P^{\ast} u 	
 			= \sum _{j\leq 2}  
 				P ^{\ast} (\sigma^W(e^j) \nabla^W_{e_j} u) 	
 			=  P ^{\ast}( D^W u ) \,.
	\end{align*}
	
	For the canonical boundary operators 
	$A^M$ on $\p M$ and $A^W$ on $P(\p M)$ adapted to $D^M$ and $D^W$ respectively it holds again by~\eqref{eq:pullback_of_composition}
	\begin{align*}
 		2A^M P^{\ast}
 			&= \sigma^M (P  ^{\ast}\nu) \sigma^M(P ^{\ast} \xi)
 				 \nabla^M_{X} P^{\ast}
 				- \sigma^M(P ^{\ast}\xi) \nabla^M_{X}
 					 \sigma^M(P ^{\ast}\nu) P^{\ast}	\\
			&= P ^{\ast} (\sigma^W (\nu) \sigma^W(\xi) 
							\nabla^W _{P _{\ast}{X}})
				-P ^{\ast} (\sigma^W (\xi) \nabla^W _{P _{\ast}{X}} \sigma^W(\nu))
			=2 P ^{\ast} A^W \,,
	\end{align*}
	where $\nu$ and $\xi$ are the normal and tangent co-vector fields
	on the boundary $P(\p M)$ and $X$ is the dual vector field to 
	$P^{\ast}\xi$.
	 We see that $\lambda$ is an eigenvalue of $A^W$ 
	 with eigenfunction $v$ if and only if
	 it is an eigenvalue of $A^M$ with an eigenfunction $P ^{\ast }v$.
	 Hence $\dom (D^M) = P ^{\ast} \dom (D^W)$.
\end{proof}

\section{Remarks on Möbius transform} \label{ap:mobius_transform}
 Möbius transform is a mapping $Y: \Cb  \rightarrow \Cb$ of the form $Y(z) = \frac{az + b}{cz+d}$ such that $ad - bc =1$.
	Notice that it is an analytic mapping on $\Cb \setminus \{ -\frac{d}{c}\}$ whose $z$ derivative
	reads
	\begin{align}\label{eq:Mob_trafo_der}
 		\p_z Y(z) 
 			= \frac{ad - bc}{(cz + d)^2} = \frac{1}{(cz + d)^2} \,.
	\end{align}
Such transforms can be obtained by the composition of the inverse stereographic projection
from the plane to a sphere, a rotation on the sphere and stereographic projecting back to the plane.
\begin{lemma} \label{le:def_Y}
	The Möbius transform $ Y _{\omega} = P R P ^{-1}$, 
	where $P$ is the stereographic projection from the north pole $N'$
	followed by the reflection across the $x$ axis
	(see~\eqref{eq:ster_proj})
	and $R$ is the rotation on $\Sb^2$ along $\phi = const$ (\ie along  a certain meridian)
	 which maps a point
	$\omega \in \Sb^2 \setminus \{N'\}$ to the north pole $N'$
	\begin{align*}
 		\omega = 
 			\begin{pmatrix}
  				\cos \phi_0 \sin \theta_0	\\
  				\sin \phi_0 \sin \theta_0 	\\
  				\cos \theta_0
			\end{pmatrix} 	
			\mapsto R(\omega)  =N' =
			\begin{pmatrix}
  				0	\\
  				0 	\\
  				1
			\end{pmatrix}	\,, \quad 
			\theta_0 \in (0, \pi] \,, \phi_0 \in (0, 2\pi] \,,	
	\end{align*}
	has the form $Y _{\omega}(z) = \frac{az + b}{cz+d}$ with
	the matrix of coefficients
	\begin{align*}
 			\begin{pmatrix}
 				a & b 	\\
 				c & d 	
			\end{pmatrix}
			=
			\begin{pmatrix}
 				\cos \frac{\theta_0}{2}	 					&	 2 \ee ^{-\ii \phi_0} \sin \frac{\theta_0}{2} 	\\
 				- \frac{1}{2} \ee ^{\ii \phi_0} \sin \frac{\theta_0}{2}  	&	 \cos \frac{\theta_0}{2}	
			\end{pmatrix} \,,
			\quad
			\det 
			\begin{pmatrix}
 				a & b 	\\
 				c & d 	
			\end{pmatrix}
			=1 \,.
	\end{align*}
	Moreover, for the composition $Y_{\omega_1}\circ Y_{\omega_2}$ for any $\omega _{1,2} \in \Sb^2\setminus \{N'\}$, 
	the coefficients satisfy the following relations
		\begin{align}\label{eq:coef_relations}
		a =  \overline{d} \,, 
		\quad 
		b = -4 \overline{c} \,,
		\quad
		|a|^2 + 4|c|^2 = |d|^2 + \frac{1}{4}|b|^2 = 1\,.
	\end{align}
\end{lemma}
\begin{proof}	
	One can easily check that $\pm 2 \ii \ee ^{-\ii \phi_0}$ are the two fixed points of $Y _{\omega}$, 
	which with the additional conditions	
	\begin{align*}
 		ad - bc =1 \quad \text{and} \quad
 		Y _{\omega}: P(\omega)  &= 2\ee ^{-\ii \phi_0} \cot \frac{\theta_0}{2} \mapsto \infty \,,
	\end{align*}
	leads to the result
	\begin{align}\label{eq:Mob_trafo}
 		Y _{\omega}(z) = 
 		\frac{\cos \frac{\theta_0}{2} z + 2 \ee ^{-\ii \phi_0} \sin \frac{\theta_0}{2}}
 			 {- \frac{1}{2} \ee ^{\ii \phi_0} \sin \frac{\theta_0}{2}  z + \cos \frac{\theta_0}{2}} \,.
	\end{align}
	In particular, let us point out that the relations between the coefficients of the Möbius transform \eqref{eq:Mob_trafo} satisfy \eqref{eq:coef_relations}.
	For a composition of  two such Möbius transforms $Y _{\omega_1}\circ Y_{\omega_2}$
	for
	\begin{align*}
 		\omega_j = 
 			\begin{pmatrix}
  				\cos \phi_j \sin \theta_j	\\
  				\sin \phi_j \sin \theta_j 	\\
  				\cos \theta_j
			\end{pmatrix} 	\,, j =1,2 \,,
	\end{align*}
	we compute using~\eqref{eq:Mob_trafo} 
\begin{align*}
	Y _{\omega_1}\circ Y_{\omega_2}(z)  &= \frac{a z +b }{cz + d} \,, \quad \text{with}	\\
 	a &= \cos \frac{\theta_1}{2} \cos \frac{\theta_2}{2} - \ee ^{-\ii(\phi_1- \phi_2)} \sin \frac{\theta_1}{2} \sin \frac{\theta_2}{2}  \\
 	d &= \cos \frac{\theta_1}{2} \cos \frac{\theta_2}{2} - \ee ^{\ii(\phi_1- \phi_2)} \sin \frac{\theta_1}{2} \sin \frac{\theta_2}{2}   = \overline{a}\\
 	b &= 2\cos \frac{\theta_1}{2} \sin \frac{\theta_2}{2} \ee ^{-\ii \phi_2} + 2\ee ^{-\ii \phi_1} \sin \frac{\theta_1}{2} \cos \frac{\theta_2}{2} 	\\
 	c &= -\frac{1}{2} \cos \frac{\theta_1}{2} \sin \frac{\theta_2}{2} \ee ^{\ii \phi_2} - \frac{1}{2} \ee ^{\ii \phi_1} \sin \frac{\theta_1}{2} \cos \frac{\theta_2}{2}  = -\frac{\overline{b}}{4} \,.
\end{align*}
\end{proof}

In what follows the particular choice of the point $\omega$ is not important  so we will generally use the notation $Y$ 
instead of $Y _{\omega}$.
Notice that Lem.~\ref{le:sgp_isometry} implies that the tangent mapping
 $Y _{\ast}$, 
``pushforwarding'' vectors at a point $z$ to vectors at $Y(z)$,
 is an isometry on the tangent space of $(\Cb, g^W)$
 with the conformal metric 
 $g^W = W^2 g = (1+ |z|^2/4) ^{-1} g$ where $g$ is the standard metric on $\Cb$.

In the last part of this section we will find
the relation between the spinor $u$ expressed in a set of coordinates on $\Cb$
and in the coordinates which are their M\"{o}bius transform.
Let $u$ be a section of the trivial $Spin^c$ spinor bundle over $\Cb$
 and denote by $u_j(z_j)$ this section in coordinates $z_j$, $j\in \{1,2\}$.
Then we have the relation
\begin{align}\label{eq:general_trafo}
 	u_1(z_1) = \mathcal{G}(z_2) u_2(z_2) 
\end{align}
for some $\mathcal{G} \in GL(2) $.
\begin{remark}
	In fact the structure group of a 
	$Spin^c$-spinor bundle over $M$
	is the group 
	 $Spin^c(2) \coloneqq Spin(2)\times U(1)/\{\pm (1, 1)\}$, where 
	 $/\{\pm (1, 1)\}$ refers to the identification 
	of classes $[(1,1)]$ and $[(-1,-1)]$,
	and, $Spin(2)\simeq SO(2)$ is the spin group of $\Rb^2$, 
	so more precisely $\mathcal{G}\in Spin^c(2) \subset GL(2)$.
	More details on Spin and $Spin^c$ groups can be found 
	\eg in \cite{LM89, TayII}.
\end{remark}
Assume further that the coordinates are related by the Möbius
transform $Y: z_2 \mapsto  z_1 = \frac{az_2 +b}{cz_2+d}$.
Since we know how the one-forms on $\Cb$
transform under a change of coordinates, 
we can find $\mathcal{G}$
by applying relation~\eqref{eq:general_trafo} on a spinor $\sigma^W(\mathcal{T})u$.
Here $\sigma^W(\mathcal{T})$ is the Clifford multiplication 
in metric $g^W$  (see Prop.~\ref{prop:conf_changes}) by a real one form
\begin{align*}
 	\mathcal{T} = \frac{1}{2} (\overline{\tau }\widehat{\dd z} + \tau \widehat{\dd \bar{z}}) \,,
\end{align*}
where $\widehat{\dd z} = W(z) \dd z$ and similarly  
$\widehat{\dd \bar{z}} = W(z) \dd \bar{z}$ denote the orthonormal basis of one forms on $\Cb$
in metric $g^W$.
We denote by $\mathcal{T} _{j} = \Re (\overline{\tau_j} \widehat{\dd z}_j)$ 
the one form  $\mathcal{T}$ in the
bases $(\widehat{\dd z} _{j}, \widehat{\dd \bar{z}} _{j})$, $j\in \{1,2\}$
and note that 
\begin{align} \label{eq:form_trafo}
 	\tau_1(Y(z)) 
 		= \frac{W(z)}{W(Y(z)) \overline{\p_z Y(z)}} \tau_2(z) 
 		= \frac{|cz+d|^2}{(cz+d)^2} \tau_2(z)\,.
\end{align}
The second equality is a result of \eqref{eq:Mob_trafo_der}  and the relations~\eqref{eq:coef_relations} for the coefficients of a Möbius transform as
	\begin{align} \label{eq:W_ratio}
 		\frac{W(z)}{W(Y(z))}
 			&=\frac{4|cz+d|^2+ |az+ b|^2}{4+|z|^2} |cz+d| ^{-2} = |cz+d| ^{-2}\,,\quad	\text{ since }\\
 			\nonumber
		 |az + b|^2 
		  	&= |az|^2 + |b|^2 + 2\Re (az \overline{b})	
		  	= 4+|z|^2 - 4|cz+d|^2 \,.
	\end{align}
By \eqref{eq:general_trafo} (taking $\sigma^W(\mathcal{T}) u$ instead of $u$)
we now obtain 
 \begin{align*}
 	\sigma^W(\mathcal{T}_1) u_1(Y(z)) 
 		= \mathcal{G}(z) \sigma^W(\mathcal{T}_2) u_2(z)
 		=  \mathcal{G}(z) \sigma^W(\mathcal{T}_2) \mathcal{G} ^{-1}(z) u_1(Y(z)) \,.
\end{align*}
Therefore we require 
\begin{align}\label{eq:rel_sigma_to_G}
 	 \mathcal{G}(z) ^{-1} \sigma^W(\mathcal{T}_1) \mathcal{G}(z)  
 		=  \sigma^W(\mathcal{T}_2) \,.
\end{align}
Prop.~\ref{prop:conf_changes} implies
$\sigma^W(\widehat{\dd z}) = \sigma(\dd z)\,, 
\sigma^W(\widehat{\dd \bar{z}}) = \sigma(\dd \bar{z})$
	and hence by \eqref{eq:Pauli_matrices}
\begin{align*}
	\sigma^W(\mathcal{T}) = 
 	\begin{pmatrix}
		0 	& \overline{\tau} 	\\
		\tau 	& 0
	\end{pmatrix} \,.
\end{align*}
We can check that setting
\begin{align*}
 	\mathcal{G}(z) = 
 		|cz+d| ^{-1}
			\begin{pmatrix}
		 		(cz + d)		&	0	\\
		 		0			&	(\overline{cz + d})
			\end{pmatrix} \in SO(2) \,,
\end{align*}
it indeed solves \eqref{eq:rel_sigma_to_G}, as
\begin{align*}
	\mathcal{G}(z) ^{-1} \sigma^W(\mathcal{T}_1) \mathcal{G}(z) 
		=
		\begin{pmatrix}
 			0 &	\frac{(\overline{cz+d})^2}{|cz+d|^2} \overline{\tau_1}	\\
 			\frac{{cz+d}^2}{|cz+d|^2} \tau_1 	& 	0
		\end{pmatrix} 
\end{align*}
corresponds to the correct transformation~\eqref{eq:form_trafo} of the components of the one form $\mathcal{T}$
establishing the equality between the right-hand side and $\sigma^W(\mathcal{T}_2)$.

For a reference we write the transformation relation for spinors on $\Cb$ under the 
Möbius transform once more with the particular 
form of $\mathcal{G}(z)$
\begin{align} \label{eq:patching}
 	u_1(z_1) = |cz_2+d| ^{-1}
			\begin{pmatrix}
		 		(cz_2 + d)		&	0	\\
		 		0			&	(\overline{cz_2 + d})
			\end{pmatrix}u_2(z_2) \,.
\end{align}

\bibliographystyle{plain}
\bibliography{../../bib/bibliography} 

\end{document}